\documentclass[11pt]{article}
\usepackage{geometry}
\usepackage[round]{natbib}

\usepackage{amsthm}

\usepackage{ifthen}

\newcounter{cdc}
\usepackage{graphicx} 
\usepackage{Shorthands}
\usepackage{authblk}
\usepackage{cleveref}
\usepackage{thmtools} 
\usepackage{thm-restate}
\usepackage{mathrsfs}  

\title{The Fundamental Limitations of Learning Linear-Quadratic Regulators}

\author[1]{Bruce D. Lee}
\author[1]{Ingvar Ziemann}
\author[2]{Anastasios Tsiamis}
\author[3]{Henrik Sandberg}
\author[1]{Nikolai Matni}
\affil[1]{Department of Electrical and Systems Engineering, University of Pennsylvania}
\affil[2]{Automatic Control Laboratory, ETH Zurich}
\affil[3]{Division of Decision and Control Systems, KTH Royal Institute of Technology}
\date{}    

\begin{document}

\maketitle

\begin{abstract}
We present a local minimax lower bound on the excess cost of designing a linear-quadratic controller from offline data. The bound is valid for any offline exploration policy that consists of a stabilizing controller and an energy bounded exploratory input. The derivation leverages a relaxation of the minimax estimation problem to Bayesian estimation, and an application of Van Trees' inequality.  We show that the bound aligns with system-theoretic intuition. In particular, we demonstrate that the lower bound increases when the optimal control objective value increases. We also show that the lower bound increases when the system is poorly excitable, as characterized by the spectrum of the controllability gramian of the system mapping the noise to the state and the $\calH_\infty$ norm of the system mapping the input to the state. We further show that for some classes of systems, the lower bound may be exponential in the state dimension, demonstrating exponential sample complexity for learning the linear-quadratic regulator offline.  
\end{abstract}
\section{Introduction}
\label{s: introduciton}

Reinforcement Learning (RL) has demonstrated success in a variety of domains, including robotics \citep{levine2016end} and games \citep{silver2017mastering}. However, it is known to be very data intensive, making it challenging to apply to complex control tasks. This has motivated efforts by both the machine learning and control communities to understand the statistical hardness of RL in analytically tractable settings, such as the tabular setting \citep{azar2017minimax} and the linear-quadratic control setting \citep{simchowitz2020naive}. 
%Does this make sense to claim both fundamental limits and efficiency of particular algorithmss?
Such studies provide insights into the fundamental limitations of RL, and the efficiency of particular algorithms. %\Ingvar{I wouldn't pair dean et al. with fundamental limitations. Their bounds suggest that easy to control => easy to learn but no converse statements. }

% \Bruce{is it a problem to refer to tabular stuff without having a related work section on it? Does Simchowitz make sense here?}
% \Bruce{rewrite with Tasos' suggestions}

There are two common problems of interest for understanding the statistical hardness of RL from the perspective of learning a linear-quadratic regulator (LQR): online LQR, and offline LQR. Online LQR models an interactive problem in which the learning agent attempts to minimize a regret-based objective, while simultaneously learning the dynamics \citep{abbasi2011regret}.  Offline LQR  models a two-step pipeline, where data from the system is collected, and then used to design a controller \citep{dean2019sample}. Guarantees in the online setting are in the form of regret bounds, whereas the offline setting focuses on Probably Approximately Correct (PAC) guarantees. The high data requirements of RL often render offline approaches the only feasible option for physical systems \cite{levine2020offline}. Despite this fact, recent years have seen greater efforts to provide lower bounds for the online LQR problem \citep{ziemann2022regret, tsiamis2022statistical}. Meanwhile, lower bounds in the offline LQR setting are conspicuously absent. Motivated by this fact, we derive lower bounds for designing a linear-quadratic controller from offline data. %We show that notwithstanding the familiarity of the setting for control theorists, several interesting consequences arise from our results. One consequence is that there exist classes of systems for which the lower bound scales exponentially with the state dimension. 

%Even with the recent interest in understanding the hardness of RL,  lower bounds in the offline linear-quadratic control setting are conspicuously absent. Instead, prior work \citep{ziemann2022regret, tsiamis2022statistical} focuses on lower bounds for linear-quadratic control in the more complicated online setting. The high data requirements for the online setting often make it infeasible for physical systems, such as robotics \citep{levine2020offline}. Motivated by this fact, we derive lower bounds for designing a linear-quadratic controller from offline data. We show that notwithstanding the familiarity of the setting for control theorists, several interesting consequences arise from our results. One consequence is that there exist classes of systems for which the lower bound scales exponentially with the state dimension. 

\textbf{Notation: } The Euclidean norm of a vector $x$ is denoted by $\norm{x}$. The quadratic norm of a vector $x$ with respect to a matrix $P$ is denoted $\norm{x}_P = \sqrt{x^\top P x}$. For a matrix $A$, the spectral norm is denoted $\norm{A}$ and the Frobenius norm is denoted $\norm{A}_F$. The spectral radius of a square matrix $A$ is denoted $\rho(A)$. A symmetric, positive semidefinite matrix $A = A^\top$ is denoted $A \succeq 0$, and a symmetric, positive definite matrix is denoted $A \succ 0$. Similarly, $A \succeq B$ denotes that $A-B$ is positive semidefinite. The eigenvalues of a symmetric positive definite matrix $A \in \R^{n \times n}$ are denoted $\lambda_1(A), \dots, \lambda_n(A)$, and are sorted in non-ascending order. We also denote $\lambda_1(A) = \lambda_{\max}(A)$, and $\lambda_{n}(A) = \lambda_{\min}(A)$. For a matrix $A$, the vectorization operator $\VEC A$ maps $A$ to a column vector by stacking the columns of $A$. The kronecker product of $A$ with $B$ is denoted $A \otimes B$. Expectation and probability with respect to all the randomness of the underlying probability space are denoted $\E$ and $\bfP$, respectively. Conditional expectation and probability given the random variable $X$ are denoted by $\E[\cdot\vert X]$ and $\bfP[\cdot\vert X]$. For an event $\calG$, $\mathbf{1}_\calG$ denotes the indicator function for $\calG$.
% For an autonomous system $x_{t+1} = Ax_t$ 
For a matrix $A\in\R^{\dx \times \dx}$ and a symmetric matrix $Q \in \R^{\dx \times \dx }$, we denote the solution $P$ to the discrete Lyapunov equation, $A^\top P A - P + Q = 0$, by $\dlyap(A, Q)$.  If we also have $B \in \R^{\du \times \du}$ and   %controlled system $x_{t+1} = Ax_t + Bu_t$ 
 $R\in\R^{\du\times \du}$, $R\succ 0$, we denote the solution $P$ to the discrete algebraic Riccati equation $Q + A^\top P A - A^\top P B(B^\top P B + R)^{-1}B^\top P A = 0$ by $\dare(A, B, Q, R)$. We use the indexing shorthand $[K] := \curly{1,\dots,K}$.

% \label{s: problem formulation}
%\Bruce{Should the underlying true parameter be denoted $\theta_\star$ to distinguish from an arbitrary $\theta$? Doing so would add some complications, as in particular assumptions would need to be stated for $\theta_\star$, or for all $\theta' \in B(\theta_\star, \e)$.}

\paragraph{Problem Formulation}
Let $\theta \in \mathbb{R}^{d_\Theta}$ be an unknown parameter. We study the fundamental limitations to learning to control the following parametric system model:
\begin{equation}
\begin{aligned}
\label{eq:lds_ac}
X_{t+1} \!=\! A(\theta) X_t \!+\! B(\theta) U_t \! +\! W_t, \, X_0 \!=\! 0, % \sim \mathsf{N}(0,\Gamma_{0}), 
\quad t\!=\!0,1,\dots.
\end{aligned}
\end{equation}
% parameterized by $\theta\in \mathbb{R}^{d_\Theta}$.
The noise process $W_t$ is assumed to be \iid\  mean zero Gaussian with fixed covariance matrices $\Sigma_W\succ 0$. The matrices $A(\theta) \in \mathbb{R}^{\dx\times\dx}$ and $B(\theta)\in\mathbb{R}^{\dx\times\du}$
%, $\Gamma_0(\theta) \in \mathbb{R}^{\dy \times \dx}$ 
are known continuously differentiable functions of the unknown parameter. The system $(A(\theta), B(\theta))$ is assumed to be stabilizable. 

We assume that the learner is given access to $N \in \mathbb{N}$ experiments $(X_{0,n},\dots,X_{T-1,n}), n\in[N]$ from \eqref{eq:lds_ac} of length $T \in \mathbb{N}$. The input signal during these experiments is  
\begin{align}
    \label{eq: input}
    U_{t,n} = F X_{t,n} + \tilde U_{t,n},
\end{align}
where $F$ renders the system stable\footnote{Access to a stabilizing controller is often assumed unstable system identification \cite{ljung1998system}. Open-loop unstable identification leads to poor conditioning.}, i.e. $\rho(A(\theta)+B(\theta) F) <1$. Meanwhile, $\tilde U_{t,n}$ is an exploration component with energy budget $ \sigma_{\tilde u}^2 NT$,\footnote{The choice to place a budget on the exploratory input $\tilde U_{t,n}$ rather than the total input $U_{t,n}$ is for ease of exposition. The energy of the exploratory input is bounded by the total budget, which is sufficient for our bounds.} where $ \sigma_{\tilde u}\in \mathbb{R}_+$. More precisely, $\tilde U_{t,n}$ may be selected as a function of past observations $(X_{0,n},\dots,X_{t,n})$, past trajectories $(X_{0,m},\dots,X_{T-1,m}),m<n$ and possible auxiliary randomization, while being constrained to an energy budget
\begin{align}
\label{eq:budgeteq}
   \frac{1}{NT} \sum_{n=1}^N \sum_{t=0}^{T-1} \E_\theta \tilde U_{t,n}^\top  \tilde U_{t,n} \leq \sigma_{\tilde u}^2 .
\end{align}
 This formulation allows both open- and closed-loop experiments, but normalizes the average exploratory input energy to $ \sigma_{\tilde u}^2$. The subscript $\theta$ on the expectation denotes that the system is rolled out with parameter $\theta$. For a fixed parameter $\theta$, we denote the data collected from these experiments by the random variable $\calZ := \{\{(X_{t,n}, U_{t,n}\}_{t=0}^{T-1}\}_{n=1}^N$. %For a fixed parameter $\theta$, we denote the joint distribution over $\{X_{t,n}\}, t=0,\dots,T-1, n \in [N]$ by $\mathsf{P}_\theta$ and the data collected from these experiments by $\calZ := \{\{(X_{t,n}, U_{t,n}\}_{t=0}^{T-1}\}_{n=1}^N$. 

% \Bruce{Added some to this explanation to make the notation consistent. Seems a bit overcomplicated. Suggestions welcome.}
The learner deploys a policy $\pi$ which is a measureable function of the $N$ offline experiments and the current state.  %, but not of the current rollout. %The current rollout is an evaluation rollout, where the learned policy is deployed in the face of new randomness. This new randomness, denoted by the random variable $\calW = \curly{W_t}_{t=0}^{T-1}$, in conjunction with the learned policy $\pi$ induces an evaluation state/input trajectory $\curly{(X_t, U_t)}_{t=0}^{T-1}$. 
In particular, the learner maps the offline data and the current state to the control input, $U_t = \pi(X_t; \calZ)$. This is the case if the learner outputs a non-adaptive state feedback controller designed with the offline data. The goal of the learner is to minimize the cost defined by: % the evaluation trajectory: %\Bruce{Explain the evaluation roll-out here, and assign a random variable that captures the randomness in the evaluation roll-out. }
\begin{align*}
V_T^{\pi}(\theta) \!&:= \!\frac{1}{T}\! \E_\theta^\pi \! \left[ \sum_{t=0}^{T-1}\! \left(  X_t^\T Q X_t +U_t^\T R U_t\right) \!+\! X_T^\T Q_T(\theta) X_T
\right]\!.
\end{align*}
The expectation is over both the offline experiments, and a new evaluation rollout. Single subscripts on the states and actions, $X_t$ and $U_t$, refer to the evaluation rollout at time $t$. %while double subscripts $X_{t,n}, U_{t,n}$ refer to the $n^\textrm{th}$ offline rollout at time $t$. 
The superscript on the expectation denotes that the inputs applied in the evaluation rollout follow the policy $U_t= \pi(X_t; \calZ)$. Note that due to the dependence of the terminal cost $Q_T(\theta)$ on the unknown parameter $\theta$, the learner does not explicitly know the cost function it is minimizing. This is not an issue: it simply means that the learner must infer the objective function from the collected data.

The following assumption guarantees the existence of a static state feedback controller that minimizes $V_T^\pi(\theta)$.
\begin{assumption}
    We assume $(A(\theta), B(\theta))$ is stabilizable, $(A(\theta), Q^{1/2})$ is detectable, and $R \succ 0$ and that $Q_T(\theta) = P(\theta)$, where $P(\theta) = \dare(A(\theta),B(\theta),Q,R)$.
\end{assumption}
 Under this assumption, the optimal policy for the known system is $U_t = K(\theta) X_t$, where $K(\theta)$ is the LQR:
\begin{align*}
    K(\theta) = -(B(\theta)^\top P(\theta) B(\theta) + R)^{-1} B(\theta)^\top P(\theta) A(\theta). 
\end{align*}
In light of this, we focus on the case in which the search space of the learner is the class of linear time-invariant state feedback policies where the gain is a measurable function of the past $N$ experiments\footnote{This assumption is not critical, and may be removed without significantly changing the result. See the proof of the main result in \cite{ziemann2022regret} for details on how to remove this assumption.}. This set is denoted $\Pi_{\mathsf{lin}}$.

The stochastic LQR cost $V_T^\pi(\theta)$ may be represented in terms of the gap between the control actions taken by the policy $\pi$ and the optimal policy, as shown below.
\begin{lemma}[Lemma 11.2 of \cite{soderstrom2002discrete}]
    \label{lem: stochastic lqr representation}
    We have that
    \[
        V_T^\pi(\theta) = \trace(P(\theta) \Sigma_W) + \frac{1}{T}\sum_{t=0}^{T-1}\!\E_{\theta}^\pi \norm{U_t- K(\theta)X_t }^2_{\Psi(\theta)},
    \]
    where $\Psi(\theta) := B^\T(\theta) P(\theta)B(\theta)+R$.
\end{lemma}

Using the above lemma, the objective of the learner may be restated from minimizing $V_T^\pi(\theta)$ to minimizing the excess cost: %suboptimality gap: %\footnote{Again assuming $Q_T(\theta) = P(\theta)$.}: %{\color{red} This is not expressed correctly}
\begin{equation}
\begin{aligned}
\label{eq: suboptality quadratic form}
% \textrm{\resizebox{\linewidth}{!}{ $\displaystyle \mathsf{EC}_T^\pi(\theta) =V_T^\pi(\theta)-\inf_{\hat \pi} V_T^{\hat \pi}(\theta)=\sum_{t=0}^{T-1} \E_{\theta}^\pi \norm{U_t- K(\theta)X_t }^2_{\Psi(\theta)},$}}
\! \mathsf{EC}_T^\pi(\theta) \! &:=\! V_T^\pi(\theta) \!-\!\inf_{\hat \pi} V_T^{\hat \pi}(\theta)\!=\!\frac{1}{T}\!\sum_{t=0}^{T-1}\!\E_{\theta}^\pi \! \norm{U_t\!-\! K(\theta)X_t }^2_{\Psi(\theta)}\!.
% \mathsf{EC}_T^\pi(\theta) &=V_T^\pi(\theta)-\inf_{\hat \pi} V_T^{\hat \pi}(\theta)\\&
% =\sum_{t=0}^{T-1} \E_{\theta}^\pi (U_t- K(\theta)X_t )^\T \Psi(\theta)  (U_t-K(\theta)X_t ),
\end{aligned}
\end{equation}
%where $\Psi(\theta) = B^\T(\theta) P(\theta)B(\theta)+R$. 
%\Tasos{I suggest to state the lemma again here without proof since it is very fundamental.} 
The second equality follows from the representation of the stochastic LQR cost in \Cref{lem: stochastic lqr representation} by cancelling the constant terms. Note that the infimum in the second term is given access to the true parameter value $\theta$, and will therefore be attained by the optimal LQR controller. In particular, it does not rely upon the offline experimental data. We denote this optimal policy by $\pi_{\theta}(X_t; \calZ) = K(\theta) X_t$.   
%\Ingvar{need to reference thm statement in söderström. Its sth like ch 11 or 12---What is "the expression"?} 

Our objective is to lower bound the excess cost for any learning agent in the class $\Pi_{\mathsf{lin}}$. To this end, we introduce the $\e$-local minimax excess cost:
%\Bruce{mathscript - mathsrf pkg or cal. }
\begin{align}
    \label{eq: minimax suboptimality}
    \mathcal{EC}^{\mathsf{lin}}_T(\theta,\e) := \inf_{\pi \in \Pi_{\mathsf{lin}}} \sup_{\|\theta'-\theta\|\leq \e} \mathsf{EC}_T^\pi(\theta').
\end{align}
To motivate this choice,
first note that if we were instead interested in an excess cost bound for only a single value of $\theta$ that holds for all estimators, the optimal policy would trivially be the LQR, $\pi(X_t, \calZ) = K(\theta)X_t$. This policy would result in a lower bound of zero. By instead requiring that the learner perform well on all parameter instances in a nearby neighborhood, we remove the possibility of the trivial solution, and can achieve meaningful lower bounds. The emphasis of the nearby neighborhood in \eqref{eq: minimax suboptimality} is essential. As the local neighborhood defined by the ball of radius $\e$, $\calB(\theta,\e) = \curly{\theta'\vert \norm{\theta'-\theta}\leq \e}$, becomes sufficiently small,  we are still able to provide instance-specific lower bounds for a single parameter value $\theta$. Therefore, the $\e$-local minimax excess cost is a much stronger notion than the standard \emph{global} minimax excess cost,  $\inf_{\pi \in \Pi_{\mathsf{lin}}} \sup_{\theta'} \mathsf{EC}_T^\pi(\theta')$, as it does not require our estimator to perform well on \emph{all possible} parameter values but only those in a small (possibly infinitesimal) neighborhood. Indeed, the global minimax excess cost for learning the optimal controller of the class of unknown stable scalar systems is infinite, as shown in \Cref{cor: global minimax}, and illustrated in \Cref{fig: hard to control instance}.

\begin{figure}
    \centering
    \ifnum\value{cdc}>0{ 
    \includegraphics[width=0.35\textwidth]{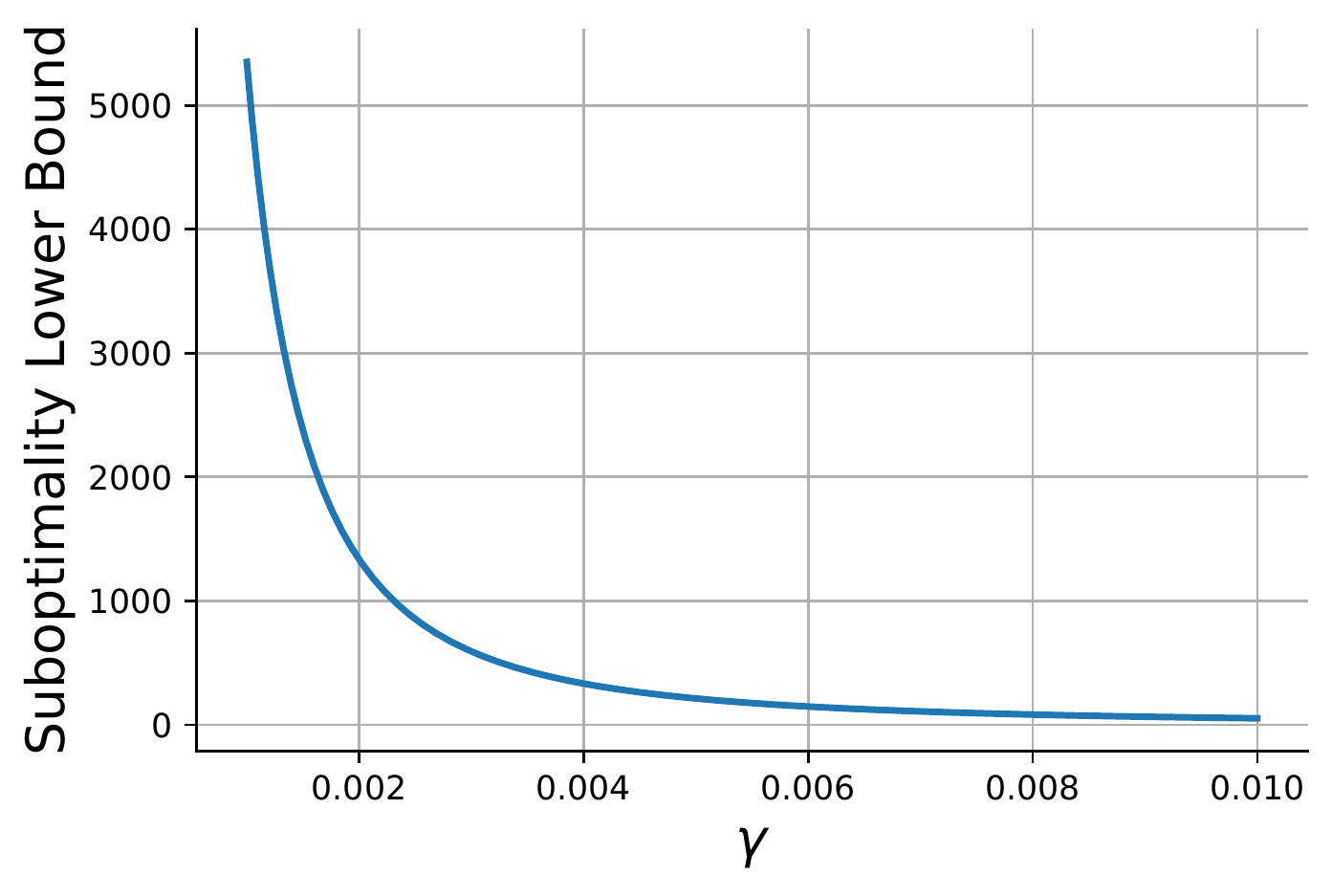}
    }\else{
     \includegraphics[width=0.5\textwidth]{figures/lower_bound_minimax_complexity.pdf}
    }\fi
    \vspace{-0.1in}
    \caption{Consider the scalar system $X_{t+1} = a X_t + b U_t + W_t$. We plot a lower bound arising from \Cref{cor: asymptotic lower bound} letting $V = \bmat{0 & 1}^\top$ for the system $a = 1- \gamma$, $b = \gamma$ as $\gamma$ ranges from $10^{-3}$ to $10^{-2}$ with $F = 0$, $ \sigma_{\tilde u}^2 =1$, $\Sigma_W = 1$. As $\gamma \to 0$ the optimally regulated system approaches marginal stability and  controllability is lost. The problem of learning a controller therefore becomes challenging as $\gamma \to 0$, which is reflected by our excess cost lower bound; it approaches $\infty$. This illustrates the observation that systems which are difficult to control are also difficult to learn to control. This plot illustrates the result \Cref{cor: global minimax} in demonstrating that the global minimax excess cost is uninformative. 
    }
    \ifnum\value{cdc}>0{\vspace{-.25in}}\else{\vspace{-0.15in}}\fi
    \label{fig: hard to control instance}
\end{figure}

Our focus in obtaining the lower bound on $\mathcal{EC}^{\mathsf{lin}}_T(\theta,\e)$ is to gain an understanding of what system-theoretic quantities render the learning problem statistically challenging. To this end, our lower bound depends on familiar system-theoretic quantities, such as $P(\theta)$. The covariance of the state under the optimal LQR controller also appears in our analysis. Under the optimal LQR controller, the covariance of the state converges to the stationary covariance as $T \to \infty$: %  \Tasos{do you need assumptions for stationarity, e.g. large $T$ or stationary of $\E X_0X_0^\top$?}:
 \ifnum\value{cdc}>0{
 \begin{align*}
     \Sigma_X(\theta) &= \lim_{t \to \infty} \E^{\pi_\theta}_{\theta} \brac{X_t X_t^\top}  \\
        &= \dlyap(\paren{A(\theta) + B(\theta) K(\theta)}^\top, \Sigma_W).
 \end{align*}
 }\else{
\[
    \Sigma_X(\theta) := \lim_{t \to \infty} \E^{\pi_\theta}_{\theta} \brac{X_t X_t^\top} =  \dlyap(\paren{A(\theta) + B(\theta) K(\theta)}^\top, \Sigma_W).  
\]
}\fi
% This quantity appears in our analysis.
 %The superscipt notation here has been overloaded with a matrix $K(\theta)$, which we take to mean that the rollout occurs with $U_t= K(\theta) X_t$. 

%The controllability gramian from the exploratory input the state during offline exploration is given by $\sum_{t=0}^\infty (A(\theta) + B(\theta) F)^t  BB^\top \paren{(A(\theta) + B(\theta) F)^t}^\top$. It does not appear directly, however the related quantity $\sum_{t=0}^\infty \norm{(A+BF)^t B}$ does appear. \Bruce{The relationship needs to be explained at some point.}
%For simplicity, we will assume that each rollout starts from the optimal stationary distribution:  $\Gamma_0(\theta) = (A(\theta) + B(\theta) K(\theta)) \Gamma_0(\theta) (A(\theta) + B(\theta) K(\theta))^\top + \Sigma_W$.
% In this paper, we lower-bound the $\e$-local minimax suboptimality gap (over a linear class $\Pi_{\mathsf{lin}}$):
% \begin{align*}
%     \mathfrak{R}^{\mathsf{lin}}_T(\theta,\e) = \inf_{\pi \in \Pi_{\mathsf{lin}}} \sup_{\|\theta'-\theta\|\leq \e} \mathsf{EC}_T^\pi(\theta').
% \end{align*}

% {\color{red} Why local minimax suboptimality?
% \begin{itemize}
%     \item global minimax complexity is infinity. 
% \end{itemize}
% }

\subsection{Contributions}
%We provide a local minimax suboptimality lower bound which scales with interpretable system-thetoretic quantities. Specifically,  

    Our main contribution is the following theorem. For the formal statements, see \Cref{thm: finite data bound} and \Cref{cor: asymptotic lower bound}. 
\begin{theorem}[Main result, Informal]
\label{thm: main result informal}
    The $\e$-local minimax excess cost is lower bounded as
    \[
        \textrm{excess cost} \geq  \frac{\textrm{system-theoretic condition number}}{\textrm{\# data points} \times \textrm{signal-to-noise ratio}}.
    \]  
\end{theorem}

In the above bound, the system-theoretic condition number depends on familiar system-theoretic quantities such as the covariance of the state under the optimal controller, and the solution to the Riccati equation. The signal-to-noise ratio depends on how easily the system is excited via both the exploratory input, and the noise. This signal-to-noise ratio may be quantified in terms of the controllability gramian of the system, as well as the exploratory input budget.
%In the above bound, $N$ is the number of offline trajectories, $\Sigma_X(\theta)$ is the stationary covariance of the state under the optimal LQR controller, $P(\theta)$ is the solution to the corresponding Riccati equation, and $D_\theta \VEC K(\theta)$ is the derivative of the controller with respect to the underlying paremeters. As $P(\theta)$ is related to optimal objective value of the stochastic LQR problem\footnote{The optimal objective value for the stochastic LQR problem is $T \trace(P(\theta) \Sigma_W)$.}, its appearance suggests that the suboptimality bound is large when the corresponding control problem is difficult. Similar intuition holds for the appearance of the state covariance, $\Sigma_X(\theta)$. The derivative of the controller $D_{\theta} \VEC K(\theta)$ captures the fact that the gap will be largest when the sensitivity of the controller to the underlying parameter is high. 

% We lower-bound the $\e$-local minimax suboptimality gap : %\Bruce{removed eq}
% \begin{align*}
%     % \label{eq: minimax suboptimality}
%     \mathfrak{R}^{\mathsf{lin}}_T(\theta,\e) \geq \frac{c(\theta,\e)}{N},
% \end{align*}
% where $c(\theta, \e)$ is a constant that captures the difficulty of the problem in terms of familiar quantities such as the solution to the Riccati equation. 

%Note that $R_T^\pi (\theta)$ is the suboptimality of an entire length $T$ roll-out. Therefore, the above bound indicates the time averaged suboptimality is inversely proportional to the total amount of offline data. 

We also study several consequences of the above result by restricting attention to the setting where all system parameters are unknown, i.e. $\VEC \bmat{A(\theta) & B(\theta)} = \theta$.  In this setting, \Cref{thm: main result informal} may be reduced to  $\mathcal{EC}^{\mathsf{lin}}_T(\theta,\e) \geq \frac{c(\theta,\e)}{NT}$, where $c(\theta,\e)$ is easily interpretable. In particular, we may reach the following conclusions: \\
\vspace{-8pt}
\begin{blockquote}
    $\bullet$ For classes of system where the operator norm of system-theoretic matrices such as the controllability gramian and the solution to the Riccati equation are constant with respect to dimension, we may take
    $c(\theta, \e) \propto \du \dx.$
    \ifnum\value{cdc}>0{
    We demonstrate that this is the optimal dependence on system dimension for this class of 
    systems. 
    }\else{Combining results from \cite{mania2019certainty} and \cite{tu2022learning} demonstrates that when $\du \leq \dx$, the upper bound on the excess cost is also proportional to $\frac{\dx \du}{NT}$. In particular, our bound is optimal in the dimension for underactuated systems when the remaining system-theoretic quantities are constant with respect to dimension.
     }\fi
     \\
    %\item 
    $\bullet$ There exist classes of systems for which we may take $c(\theta, \e) \propto \exp(\dx)$. %4^{\dx}\Tasos{\mathrm{exp(\dx)}?}.$ 
    This demonstrates that the excess cost of a learned LQR controller may grow exponentially in the dimension. \\
    %\item 
    $\bullet$ 
    \ifnum\value{cdc}>0{We may take $c(\theta, \e)$ to grow with iterpretable system-theoretic quantities, such as the eigenvalues of both the solution to the Riccati equation, $P(\theta),$ and the state covariance under the optimal controller, $\Sigma_X(\theta)$.}
    \else{The lower bound grows in an interpretable manner with familiar system-theoretic quantities. In particular, we may take $c(\theta, \e)$ to grow with the eigenvalues of both the solution to the Riccati equation, $P(\theta),$ and the state covariance under the optimal controller, $\Sigma_X(\theta)$. }\fi 
    This suggests that the problem of learning to control a system with a small gap from the optimal controller is data intensive when controlling the underlying system is hard. %The constant $c(\theta,\e)$ also scales inversely with an upper bound on a measure of the ease of identification from the offline experiments. In particular, $c(\theta,\e)$ grows large when it becomes difficult to identify the system due to poor excitability by the noise and the exploratory input. 
    %\Tasos{intensive if objective is small sub-opt gap}
\end{blockquote}

\subsection{Related Work}

\begin{figure*}[t]
    \centering
    \vspace{-24pt}
    \includegraphics[width=0.8\textwidth]{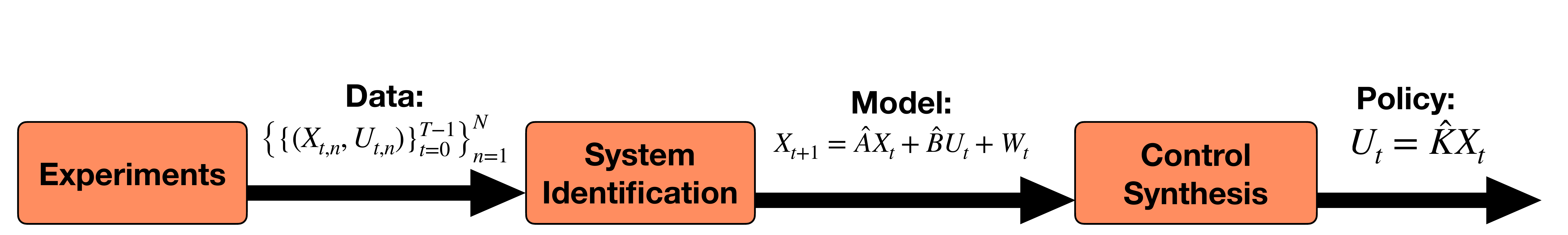}
    \vspace{-0.1in}
    \caption{A classic model-based pipeline for learning a controller from data. }%\Tasos{couldn't find a reference to the fig. in the text?}}%\Bruce{RVs Capital, fix spacing}}
    \label{fig: control from data pipeline}
    \ifnum\value{cdc}>0{\vspace{-.25in}}\else{\vspace{-0.15in}}\fi
\end{figure*}

\paragraph{System Identification} System identification is often a first step in designing a controller from experimental data, and has a longstanding history. The text \cite{ljung1998system} covers classic asymptotic results. Control oriented identification was studied in \cite{ chen1993caratheodory, helmicki1991control}.  Recently, there has been interest in finite sample analysis for fully-observed linear systems \citep{dean2019sample, simchowitz2018learning, faradonbeh2018finite, sarkar2019near}, and partially-observed linear systems \citep{oymak2019non, sarkar2021finite, simchowitz2018learning, tsiamis2019finite, lee2020non, zheng2020non}. Lower bounds for the sample complexity of system identification are presented in \cite{jedra2019sample, tsiamis2021linear}. For a more extensive discussion of prior work, we refer to  the survey by \cite{tsiamis2022statistical}.

\paragraph{Learning Controllers Offline} 
 %\Bruce{Make concise}
Learning a controller from offline data is a familiar paradigm for control theorist and practitioners. It  typically consists of system identification, followed by robust \citep{zhou1996robust} or certainty-equivalent \citep{simon1956dynamic} control design, see \Cref{fig: control from data pipeline}. %\Ingvar{probably should cite feldbaum or simon for CE}.  
Recent work provides finite sample guarantees for such methods \citep{dean2019sample, mania2019certainty}. Upper and lower bounds on the sample complexity of stabilization from offline data are presented in \cite{tsiamis2022learning}. The RL community has a similar paradigm, known as offline RL \citep{levine2020offline}. Policy gradient approaches are a model-free algorithm suitable for offline RL, and are  analyzed in \cite{fazel2018global}. Lower bounds on the variance of the gradient estimates in policy gradient approaches are supplied in \cite{ziemann2022policy}. Lower bounds for offline linear control are also studied in \cite{wagenmaker2021task} with the objective of designing optimal experiments. We instead focus on the LQR setting to understand the dependence of the excess cost on interpretable system-theoretic quantities. %This in turn provides an understanding of the statistical difficulty of learning controllers. 

% the best of our knowledge, lower bounds on the suboptimality of a LQR learned offline do not exist. The lower bounds presented in this work are valid for any model-based or model-free approach to learn the controller from offline data.  

\paragraph{Online LQR} The problem of learning the optimal LQR controller online has a rich history beginning with 
\cite{aastrom1973self}.  Regret minimization was introduced in \cite{lai1986asymptotically, lai1986extended}. The study of regret in online LQR was re-initiated by
\cite{abbasi2011regret}, inspired by works in the RL community. Many works followed to propose algorithms which were computationally tractable 
\citep{ouyang2017control,dean2018regret, abeille2018improved, mania2019certainty, cohen2019learning, faradonbeh2020input, jedra2021minimal}. Lower bounds on the regret of online  LQR are presented in
\cite{simchowitz2020naive, cassel2020logarithmic, ziemann2022regret}. The results in this paper follow a similar proof to \cite{ziemann2022regret}. The primary difference is that since our controller is designed via offline data, we may not make use of the exploration-exploitation tradeoff to upper bound the information available to the learner, as is done in \cite{ziemann2022regret}.
%Unlike these works, this paper provides lower bounds for \emph{offline LQR}. 

%\Ingvar{I think we probably will need to explain relationship to Ziemann Sandberg 2022 in a little more detail since the proof approach is kinda taken from our paper. I can get on this in a week or two tho}

% \input{problem_formulation}
\section{Excess Cost Lower Bound}
\label{s: suboptimality lower bounds}
% Before presenting our lower bounds, we  motivate our choice of the $\e$-local minimax suboptimality gap in \eqref{eq: minimax suboptimality} as the relevant quantity to study.   

We now proceed to establish our lower bound.  \ifnum\value{cdc}>0{Missing proofs and further details may be found in \cite{lee2023fundamental}. }\else{}\fi As we are interested in the worst-case excess cost from any element of $\calB(\theta,\e)$, we make the additional assumption that $F$ stabilizes $(A(\theta'), B(\theta'))$ for all $\theta'\in \calB(\theta,\e)$.\footnote{We ultimately study the limit as $\e$ becomes small. Therefore, this is not significantly stronger than assuming that $F$ stabilizes ($A(\theta), B(\theta)$).} %\Bruce{and that F stabilizes all such instances?} \Ingvar{Yeah and that $F$ is stabilizing in this nbhd. idk where we should state this though! probably further down closer to the corollary} \Bruce{I think it's actually necessary here with how I'm writing the lower bounds. Assumes that $A(\theta') + B(\theta') F$ is stable for all $\theta' in \calB(\theta,\e)$. }
This also ensures that the optimal LQR controller exists for all points in the prior.

%\Tasos{My suggestion: Start with the main result/main theorem (e.g. th 2.2 with cleaner notation), along with interpretation. Then explain the proof.}
To obtain a lower bound on the local minimax excess cost, we lower bound the maximization over $\theta' \in \calB(\theta,\e)$ by an average over a distribution supported on $\calB(\theta, \e)$. This reduces the problem to lower bounding a Bayesian complexity. Instead of fixing the parameter $\theta$, we let $\Theta$ be a random vector taking values in $\mathbb{R}^{d_\Theta}$ and suppose that it has prior density $\lambda$. Doing so enables the use of information theoretic tools to lower bound the complexity of estimating the parameter from data. The relaxation of the the maximization is shown in the following lemma. 

\begin{lemma} 
\label{lem: weak duality lb}
Fix $\e> 0$ and let $\lambda$ be any prior on  $\calB(\theta,\e)$. Then for any $\pi \in \Pi^\mathsf{lin}$ with $\pi(X_t, \calZ) = \hat K(\calZ)X_t$,  
%$\mathfrak{EC}^{\mathsf{lin}}_T(\theta,\e)$ is lower bounded by %If the costs are stationary:
\ifnum\value{cdc}>0{
\begin{align*}
    %&\mathfrak{EC}^{\mathsf{lin}}_T(\theta,\e)
    &\sup_{\theta' \in \calB(\theta)} \mathsf{EC}_T^\pi(\theta') \geq \\ 
    & \E  \tr  \paren{[\hat K(\mathcal{Z}) - K(\Theta)])^\T \Psi(\Theta) [\hat K(\mathcal{Z})-K(\Theta)] \Sigma_{\Theta}^{\hat K(\calZ)} }, 
\end{align*}
}\else{\begin{align*}
    %&\mathfrak{EC}^{\mathsf{lin}}_T(\theta,\e)
    &\sup_{\theta' \in \calB(\theta)} \mathsf{EC}_T^\pi(\theta')  \geq \E  \tr  \paren{[\hat K(\mathcal{Z}) - K(\Theta)])^\T \Psi(\Theta) [\hat K(\mathcal{Z})-K(\Theta)] \Sigma_{\Theta}^{\hat K(\calZ)} }, 
\end{align*}}\fi
where $\Sigma_{\Theta}^{\hat K(\calZ)} := \frac{1}{T} \sum_{t=0}^{T-1} \E^{\pi} \brac{X_tX_t^\T \vert \calZ, \Theta}$. The expectation is over the prior $\Theta\sim \lambda$, and the randomness of both the offline rollouts and the evaluation rollout. We recall the shorthand $\Psi(\Theta) = B(\Theta)^\top P(\Theta) B(\Theta) +R$. 
%We recall the shorthand $\Psi(\Theta) = B^\T(\Theta) P(\Theta)B(\Theta)+R$.
\end{lemma}
\ifnum\value{cdc}>0{}\else{
\begin{proof}
By the quadratic expression for the excess cost in \eqref{eq: suboptality quadratic form} and the fact that the supremum over a set always exceeds the weighted average over a set, 
we have the following inequality:
% \ifnum\value{cdc}>0{
\begin{align*}
    %&\mathfrak{EC}^{\mathsf{lin}}_T(\theta,\e)
    &  \sup_{\|\theta'-\theta\|\leq \e} \mathsf{EC}_T^\pi(\theta') \geq \underset{\Theta\sim\lambda}{\E} R_T^{\pi}(\Theta) \\
    &\!\!=\! %\inf_{\pi \in \Pi_{\mathsf{lin}}} 
    \frac{1}{T}\E \sum_{t=0}^{T-1} \E^{\pi} \brac{\norm{U_t \!-\! K(\Theta)X_t}_{\Psi(\Theta)}^2 \vert \calZ, \Theta} \\ 
    &\!\!=\! %\inf_{\hat K}
    \E\tr ([\hat K(\mathcal{Z})\!-\!K(\Theta)])^\T 
      \Psi(\Theta)([\hat K(\mathcal{Z})-K(\Theta)] \Sigma_{\Theta}^{\hat K\!(\!\calZ\!)} ).
   % &\mathfrak{EC}^{\mathsf{lin}}_T(\theta,\e) = \inf_{\pi \in \Pi_{\mathsf{lin}}} \sup_{\|\theta'-\theta\|\leq \e} \mathsf{EC}_T^\pi(\theta') \geq \inf_{\pi \in \Pi_{\mathsf{lin}}} \underset{\Theta\sim\lambda}{\E} R_T^{\pi}(\Theta) \\
   %  &= \inf_{\pi \in \Pi_{\mathsf{lin}}} \E \sum_{t=0}^{T-1} \E_{\calZ, \Theta}^{\pi} (U_t- K(\Theta)X_t )^\T \\ &\qquad \cdot\Psi(\Theta) (U_t-K(\Theta)X_t )\\ 
   %  &= \inf_{\hat K} \E\tr  ([\hat K(\mathcal{Z})-K(\Theta)])^\T 
   %   \\ &\qquad\cdot \Psi(\Theta)([\hat K(\mathcal{Z})-K(\Theta)] T \Sigma_{\Theta}^{\hat K(\calZ)} ).
\end{align*}
% }\else{
% \begin{align*}
%    &\mathfrak{EC}^{\mathsf{lin}}_T(\theta,\e) = \inf_{\pi \in \Pi_{\mathsf{lin}}} \sup_{\|\theta'-\theta\|\leq \e} \mathsf{EC}_T^\pi(\theta') \geq \inf_{\pi \in \Pi_{\mathsf{lin}}} \underset{\Theta\sim\lambda}{\E} R_T^{\pi}(\Theta) \\
%     &= \inf_{\pi \in \Pi_{\mathsf{lin}}} \E \sum_{t=0}^{T-1} \E_{\calZ, \Theta}^{\pi} (U_t- K(\Theta)X_t )^\T \Psi(\Theta) (U_t-K(\Theta)X_t )\\ 
%     &= \inf_{\hat K} \E\tr  ([\hat K(\mathcal{Z})-K(\Theta)])^\T 
%     \Psi(\Theta) ([\hat K(\mathcal{Z})-K(\Theta)] T \Sigma_{\Theta}^{\hat K(\calZ)} ).
% \end{align*}}\fi
The second to last equality follows by the tower rule. 
The last equality results by substituting $U_t=\hat K(\calZ) X_t$, followed by the trace-cyclic property and linearity of expectation.
\end{proof}
}\fi

We may treat the data from offline experimentation, $\calZ$, as an observation of the underlying parameter $\Theta$. In particular, $\calZ$ may be expressed as a random vector taking values in $\R^{NT(\dx+\du)}$ with conditional density $p(\cdot|\theta)$. The following Fisher information matrix and prior density concentration matrix measure estimation performance of $\Theta$ from the sample $\calZ$ with respect to the square loss:
\begin{align}
\label{eq:fisherdef}
\I_p(\theta) &:=\int  \left( \frac{\nabla_\theta p(z|\theta)}{p(z|\theta)}\right)\left( \frac{\nabla_\theta p(z|\theta)}{p(z|\theta)}\right)^\T p(z|\theta) dz, \\%\textnormal{ and} \\
\label{eq:fisherlocdef}
\J(\lambda)&:= \int \left( \frac{\nabla_\theta \lambda(\theta)}{\lambda(\theta)}\right) \left( \frac{\nabla_\theta \lambda(\theta)}{\lambda(\theta)}\right)^\T  \lambda(\theta)  d\theta.
\end{align}
The first quantity \eqref{eq:fisherdef} measures the information content of the sample $\calZ$ with regards to $\Theta$. The second quantity \eqref{eq:fisherlocdef} measures the concentration of the prior density $\lambda$. As the gradient operator $\nabla_\theta$ maps to a vector of dimension $d_\Theta$, both $I_p(\theta)$ and $J(\lambda)$ are $d_\Theta\times d_\Theta$ dimensional. See \cite{ibragimov2013statistical} for further details about these integrals and their existence.

As we seek lower bounds for estimating $K(\Theta)$ instead of just $\Theta$, we must account for the transformation from a quadratic loss over the error in esimating $\Theta$ to the error in estimating $K(\Theta)$, as appears in \Cref{lem: weak duality lb}. To do so, we introduce the Van Trees' inequality \citep{van2004detection, bobrovsky1987some}. We first impose the following standard regularity conditions:
\begin{assumption} \,
\label{asm: van trees reguarlity} 
\begin{enumerate}
\item The prior $\lambda $ is smooth with compact support. 
\item The conditional density of $\calZ$ given $\Theta$, $p(z|\cdot)$, is continuously differentiable on the domain of $\lambda$ for almost every $z$.
\item The score\footnote{The score is the gradient of the log-likelihood. It evaluates to $\frac{\nabla_\theta p(z|\theta)}{p(z|\theta)}$.} has mean zero; $ \int \left( \frac{\nabla_\theta p(z|\theta)}{p(z|\theta)}\right) p(z|\theta) dz =0$.
\item $\J(\lambda)$ is finite and $\I_p(\theta)$ is a continuous function of $\theta$ on the domain of $\lambda$.
\item $\VEC K$ is differentiable on the domain of $\lambda$.
\end{enumerate}
\end{assumption}

The following theorem is a less general adaption from \cite{bobrovsky1987some}  which suffices for our needs. 
\begin{theorem}[Van Trees Inequality]
\label{thm:vtineq}
Fix two random variables $(\calZ,\Theta) \sim   p(\cdot|\cdot) \lambda(\cdot)$  and suppose Assumption~\ref{asm: van trees reguarlity} holds. Let $\mathcal{G}$ be a $\sigma(\calZ)$-measurable event. Then for any $\sigma(\calZ)$-measurable $\hat K$:
\ifnum\value{cdc}>0{
\begin{equation}
\label{VTineq}
\begin{aligned}
&\E \left[ \VEC (\hat K(\calZ) -K(\Theta)) \VEC (\hat K (\calZ)-K(\Theta))^\T \mathbf{1}_\mathcal{G}\right] \succeq \\
&  \E[ \dop_\theta \VEC K(\Theta)\mathbf{1}_\mathcal{G}]^\T \!\left[ \E \I_p(\Theta)+\J(\lambda) \right]^{-1} \! \E [\dop_\theta \VEC K(\Theta)\mathbf{1}_\mathcal{G}].
\end{aligned}
\end{equation}
}\else{
\begin{equation}
\label{VTineq}
\begin{aligned}
&\E \left[ \VEC (\hat K(\calZ) -K(\Theta)) \VEC (\hat K (\calZ)-K(\Theta))^\T \mathbf{1}_\mathcal{G}\right] \\
&\succeq  \E[ \dop_\theta \VEC K(\Theta)\mathbf{1}_\mathcal{G}]^\T \left[ \E \I_p(\Theta)+\J(\lambda) \right]^{-1} \E [\dop_\theta \VEC K(\Theta)\mathbf{1}_\mathcal{G}].
\end{aligned}
\end{equation}}\fi
The notation $\dop_\theta \VEC K(\cdot)$ above follows the standard convention for a Jacobian: it stacks the transposed gradients of each element of $\VEC K(\cdot)$ into a $\dx \du \times d_\Theta$ dimensional matrix. 
\end{theorem}

We see from \Cref{thm:vtineq}  
that the transformation to the error in estimating $K(\Theta)$ is accounted for by $\dop_\theta \VEC K(\cdot)$. 

We now massage the lower bound in \Cref{lem: weak duality lb} to a form compatible with \Cref{thm:vtineq}. Doing so requires us to express the lower bound as a quadratic form conditioned on some $\sigma(\calZ)$-measureable event $\calG$. We therefore select an event $\calG$ for which we may uniformly lower bound the quantities $\Psi(\Theta)$ and $\Sigma_{\Theta}^{\hat K(\calZ)}$. To this end, we define positive definite matrices $\Psi_{\theta,\e}$ and $\Sigma_{\theta,\e}$ that satisfy 
\begin{align}
    \label{eq: uniform lower bound on N and Gamma} 
    \Psi(\theta') &\succeq \Psi_{\theta,\e}  \textrm{ and }  \frac{1}{2} \Sigma_X(\theta')
    \succeq \Sigma_{\theta,\e} \quad \forall \theta' \in\calB(\theta,\e).
\end{align}
The matrix $\Psi_{\theta,\e}$ will serve to uniformly lower bound $\Psi(\Theta)$. When the learned controller is close to the optimal controller, the covariance of the state under the learned controller will be close to the covariance of the state under the optimal controller, which is in turn lower bounded in terms of $\Sigma_{\theta, \e}$. In particular, if $\norm{\hat K(\calZ) - K(\Theta)}$ is sufficiently small, we can argue that $\Sigma_{\Theta}^{\hat K(\calZ)} \succeq \frac{1}{2}\Sigma_X(\Theta) \succeq \Sigma_{\theta,\e}$. The aforementioned condition on $\norm{\hat K(\calZ) - K(\Theta)}$ will hold only if there is a large amount of data available to fit $\hat K(\calZ)$. To achieve a bound that holds in the low data regime, we observe that the state covariance under the learned controller is always lower bounded by the noise covariance: $\Sigma_\Theta^\calZ \succeq \Sigma_W$.  
%\Tasos{it took me quite some time to understand the definition of Sigma ``when we are near optimal"}.
For this reason, the subsequent results will be presented in two parts: one in which we condition on an event where $\norm{\hat K(\calZ) - K(\theta)}$ is small, and one that holds generally. To present these results concisely, the positive definite matrix $\Gamma_{\theta,\e}$ is used to denote either $\Sigma_W$ or $\Sigma_{\theta,\e}$. The Kronecker product of these lower bounds arises frequently, motivating the shorthand 
\begin{align}
    \label{eq: kronecker shorthand}
    \Xi_{\theta, \e} := \Gamma_{\theta, \e} \otimes \Psi_{\theta,\e}.
\end{align}

%\Bruce{Express the following results with some condensed notation, and modify the proofs accordingly. State the assumptions on the prior.}
\begin{lemma}[Application of Van Trees' Inequality]
\label{lem: application of van trees}
For any smooth prior $\lambda$ on $\calB(\theta,\e)$ and any $\pi \in \Pi^\mathsf{lin}$ with $\pi(X_t, \calZ) = \hat K(\calZ) X_t$,  
% $\mathfrak{EC}^{\mathsf{lin}}_T(\theta,\e)$ is lower bounded by
\ifnum\value{cdc}>0{
\begin{equation}
    \begin{aligned}
    \label{eq: unified LB}
       %&\mathfrak{EC}^{\mathsf{lin}}_T(\theta,\e) 
       &\sup_{\theta'\in\calB(\theta,\e)}\mathsf{EC}_T^\pi(\theta') \geq \\ & \frac{\tr\left(   \Xi_{\theta,\e}  \E [\dop_\theta \VEC  K(\Theta) \mathbf{1}_{\calG}] \E [\dop_\theta  \VEC    K(\Theta) \mathbf{1}_{\calG}]^\T \right) }{ \norm{\E \I_p(\Theta)+\J(\lambda)}},
       \end{aligned}
\end{equation}
}\else{
\begin{align}  
    \label{eq: unified LB}
    \sup_{\theta'\in\calB(\theta,\e)}\mathsf{EC}_T^\pi(\theta') \geq \frac{\tr\left(   \Xi_{\theta,\e}  \E [\dop_\theta \VEC  K(\Theta) \mathbf{1}_{\calG}] \E [\dop_\theta  \VEC    K(\Theta) \mathbf{1}_{\calG}]^\T \right) }{ \norm{\E \I_p(\Theta)+\J(\lambda)}},
\end{align}}\fi
where either:
\begin{blockquote}
    $1)$ $\Gamma_{\theta, \e} = \Sigma_W$ and $\calG = \Omega$, \textrm{ or }\\
    $2)$ $\Gamma_{\theta, \e} = \Sigma_{\theta, \e}$ and $\calG = \calE$, if $\displaystyle T \geq \sup_{\theta'\in \calB(\theta,\e)}\frac{16 \norm{\Sigma_X(\theta')}}{\lambda_{\min}(\Sigma_X(\theta'))}$. 
\end{blockquote} 
\noindent The event $\Omega$ is the entire sample space, i.e. $\P\brac{\Omega} = 1$, and  
\ifnum\value{cdc}>0{
\begin{align*}
    \mathcal{E} &= \curly{\sup_{\theta' \in \calB(\theta,\e)} \norm{\hat K(\calZ) - K(\theta')} \leq  \alpha } \\
    \alpha &= \inf_{\theta' \in \calB(\theta,\e)} \min\bigg\{\frac{\norm{A_{cl}(\theta')}}{\norm{B(\theta')}}, \\
    &\qquad \frac{\lambda_{\min}(\Sigma_X(\theta'))/24}{\norm{A_{cl}(\theta')} \norm{B(\theta')} \mathcal{J}(A_{cl}(\theta')) \norm{\Sigma_X(\theta')}}\bigg\}.
\end{align*}
}\else{
\begin{align*}
    \mathcal{E} &= \curly{\sup_{\theta' \in \calB(\theta,\e)} \norm{\hat K(\calZ) - K(\theta')} \leq  \alpha }  \\
    \alpha &= \inf_{\theta' \in \calB(\theta,\e)} \min\curly{\frac{\norm{A_{cl}(\theta')}}{\norm{B(\theta')}}, \frac{\lambda_{\min}(\Sigma_X(\theta'))/24}{\norm{A_{cl}(\theta')} \norm{B(\theta')} \mathcal{J}(A_{cl}(\theta')) \norm{\Sigma_X(\theta')}}}.
\end{align*}
}\fi
 %\textrm{ and } \sup_{\theta' \in \calB(\theta,\e)} \rho(A(\theta') + B(\theta') \hat K(\calZ)) < 1}$, 
%and \[\.
%\]
Here, $A_{cl}(\theta) = A(\theta) + B(\theta)K(\theta)$ and $\calJ(A_{cl}(\theta)) = \sum_{t=0}^\infty \norm{A_{cl}(\theta)^t}^2$.

% Then 
% \begin{align}  
%     \label{eq: general LB}
%      \mathfrak{EC}^{\mathsf{lin}}_T(\theta,\e)  \geq    T\tr\left(   (\Sigma_W \otimes \Psi_{\theta,\e} )  \E [\dop_\theta \VEC  K(\Theta)] \left[ \E \I_p(\Theta)+\J(\lambda) \right]^{-1} \E [\dop_\theta  \VEC    K(\Theta)]^\T \right).
% \end{align}
% If we also have that $T \geq \frac{16\norm{\Gamma_0}^2}{\lambda_{\min}(\Gamma_0)}$ then 
% \begin{align}
%     \label{eq: conditional LB}
%     \mathfrak{EC}^{\mathsf{lin}}_T(\theta,\e)  \geq    \tr\left(   (\Gamma_{\theta,\e} \otimes \Psi_{\theta,\e} )  \E [\dop_\theta \VEC  K(\Theta) \mathbf{1}_\calE] \left[ \E \I_p(\Theta)+\J(\lambda) \right]^{-1} \E [\dop_\theta  \VEC    K(\Theta)\mathbf{1}_\calE]^\T \right),
% \end{align}
% where $\mathcal{E} = \curly{\sup_{\theta' \in \calB(\theta,\e)} \norm{\hat K(\calZ) - K(\theta')} \leq  \alpha }$, %\textrm{ and } \sup_{\theta' \in \calB(\theta,\e)} \rho(A(\theta') + B(\theta') \hat K(\calZ)) < 1}$, 
% and \[\alpha = \inf_{\theta' \in \calB(\theta,\e)} \min\curly{\frac{\norm{A_{cl}(\theta')}}{\norm{B(\theta')}}, \frac{\lambda_{\min}(\Gamma_0(\theta'))}{24 \norm{A_{cl}(\theta')} \norm{B(\theta')} \mathcal{J}(A_{cl}(\theta'))^2 \norm{\Gamma_0(\theta')}}}.
% \]
% Here, $A_{cl}(\theta) = A(\theta) + B(\theta)K(\theta)$.
%$\calE = \curly{R_T^{\pi(\calZ)}(\Theta) \leq \sqrt\frac{T}{N}}$.
\end{lemma}

\ifnum\value{cdc}>0{
The proof of the above result follows by applying \Cref{thm:vtineq} to the lower bound in \Cref{lem: weak duality lb} after replacing $\Psi(\Theta)$ by $\Psi_{\theta,\e}$ and $\Sigma_\Theta^{\hat K(\calZ)}$ by $\Gamma_{\theta,\e} \mathbf{1}_\calG$. 
}\else{
\begin{proof}
We always have that $\Sigma_{\hat K(\calZ)} \succeq \Sigma_W$. 
\ifnum\value{cdc}>0{Lemma A.2 of \cite{lee2023fundamental} 
}\else{\Cref{lem: covariance overlap between learned and optimal controller} }\fi  shows that if $T \geq \frac{16 \norm{\Sigma_X(\Theta)}^2}{\lambda_{\min}(\Sigma_X(\Theta))}$, then under event $\calE$,  we have $\norm{\Sigma_X(\Theta)^{-1/2}\Sigma_{\Theta}^{\hat K(\calZ)}\Sigma_X(\Theta)^{-1/2} - I} \leq \frac{1}{2}$. This in turn implies that $\Sigma_{\Theta}^{\hat K(\calZ)} \succeq \Sigma_{\theta, \e} \mathbf{1}_\calE$. 

%When the rollouts are sufficiently long such that $T \geq \frac{16 \norm{\Gamma_0}^2}{\lambda_{\min}(\Gamma_0)}$, then under event $\calE$, we can show $\Gamma_{\hat K(\calZ)}(\Theta) \succeq \frac{1}{2} \Gamma_0(\Theta)$. 
% Observe that
% \begin{align*}
%     & \inf_{\hat K} \E_{\Theta\sim \lambda} \E_{\calZ|\Theta} \tr  ([\hat K(\mathcal{Z})-K(\Theta)])^\T \Psi(\Theta) ([\hat K(\mathcal{Z})-K(\Theta)] T \Gamma_{\hat K(\calZ)}(\Theta) ) \\
%     %&= \inf_{\hat K} \E_{\Theta\sim \lambda} \E_{\calZ|\Theta} \tr  ([\hat K(\mathcal{Z})-K(\Theta)])^\T (B^\T(\Theta) P(\Theta)B(\Theta)+R) ([\hat K(\mathcal{Z})-K(\Theta)] T \Gamma_0 \Gamma_0^{-1}\Gamma_{\hat K(\calZ)}(\Theta) ) \\
%     &\geq \inf_{\hat K} \E_{\Theta\sim \lambda} \E_{\calZ|\Theta} \tr  ([\hat K(\mathcal{Z})-K(\Theta)])^\T \Psi(\Theta)  ([\hat K(\mathcal{Z})-K(\Theta)] T \Gamma_0) \sigma_{\min}(\Gamma_0^{-1}\Gamma_{\hat K(\calZ)}(\Theta)) 
% \end{align*}
%In particular \Cref{lem: covariance overlap between learned and optimal controller} shows that if $T \geq \frac{16 \norm{\Gamma_0}^2}{\lambda_{\min}(\Gamma_0)}$, then under event $\calE$,  we have $\norm{\Gamma_0^{-1/2}\Gamma_{\hat K(\calZ)}(\Theta)\Gamma_0^{-1/2} - I} \leq \frac{1}{2}$. 
% This in turn implies that $\sigma_{\min}(\Gamma_0^{-1}\Gamma_{\hat K(\calZ)}(\Theta)) \geq \frac{1}{2}$. Alternatively, we can always replace $\Gamma_{\hat K(\calZ)}$ by $\Sigma_W$ and achieve a valid lower bound. 

With this fact in hand, we may replace $\Psi(\Theta)$ in the lower bound from \Cref{lem: weak duality lb} by $\Psi_{\theta,\e}$, and  $\Sigma_{\Theta}^{\hat K(\calZ)}$ by $\Gamma_{\theta,\e} \mathbf{1}_\calG$, where $(\Gamma_{\theta,\e}, \calG)$ can only be set as $(\Sigma_{\theta,\e}, \calE)$ if $T$ is sufficiently large. 
Then 
\begin{align*} 
     &\sup_{\theta'\in\calB(\theta,\e)}\mathsf{EC}_T^\pi(\theta')  \geq    \E %\sum_{t=0}^{T-1}
     \tr  ([\hat K(\calZ)-K(\Theta)])^\T  \Psi_{\theta,\e} ([\hat K(\calZ)-K(\Theta)]  \Gamma_{\theta,\e}) \mathbf{1}_\calG  \\ 
     &= \E %\sum_{t=0}^{T-1} 
     \tr  ([\tilde K(\calZ)-   \sqrt{\Psi_{\theta,\e}} K(\Theta)\sqrt{\Gamma_{\theta,\e}}])^\T  [\tilde K(\calZ)-   \sqrt{\Psi_{\theta,\e}} K(\Theta)\sqrt{\Gamma_{\theta,\e}}] )\mathbf{1}_\calG\\ %\underset{\Theta\sim\lambda}{\E_{\Theta}}
     &=  \E \tr \big(  [\VEC \tilde K(\calZ)-   \VEC\sqrt{\Psi_{\theta,\e}} K(\Theta)\sqrt{\Gamma_{\theta,\e}}] \cdot [\VEC \tilde K(\calZ)-   \VEC \sqrt{\Psi_{\theta,\e}} K(\Theta)\sqrt{\Gamma_{\theta,\e}}] \big)^\T \mathbf{1}_\calG, %\underset{\Theta\sim\lambda}{\E_{\Theta}}
\end{align*}
where $\tilde K(\calZ) = \sqrt{\Psi_{\theta,\e}} \hat K(\calZ) \sqrt{\Gamma_{\theta,\e}}$.
% }\fi
We now invoke the Van Trees' inequality, \Cref{thm:vtineq}:
\begin{align*}
    &\sup_{\theta'\in\calB(\theta,\e)}\mathsf{EC}_T^\pi(\theta')
    \!\geq \!
      \trace\paren{ \!\E[ \dop_\theta  \VEC\sqrt{\Psi_{\theta,\e}} K(\Theta)\sqrt{\Gamma_{\theta,\e}} \mathbf{1}_\calG] \left[ \E \I_p(\Theta)\!+\!\J(\lambda) \right]^{-1}\! \E [\dop_\theta  \VEC\sqrt{\Psi_{\theta,\e}} K(\Theta)\sqrt{\Gamma_{\theta,\e}} \mathbf{1}_\calG]^\T \!}\\
     & \!= \!   \tr \paren{ \!\E[  (\sqrt{\Gamma_{\theta,\e}}\! \otimes\! \sqrt{\Psi_{\theta,\e}} )  \dop_\theta \VEC  K(\Theta) \mathbf{1}_\calG] \left[ \E \I_p(\Theta)+\J(\lambda) \right]^{-1} \E [ (\sqrt{\Gamma_{\theta,\e}} \! \otimes \! \sqrt{\Psi_{\theta,\e}} ) \dop_\theta  \VEC    K(\Theta) \mathbf{1}_\calG]^\T \!},
\end{align*}
where we  used that $\VEC\sqrt{\Psi_{\theta,\e}} K(\Theta)\sqrt{\Gamma_{\theta,\e}} = (\sqrt{\Gamma_{\theta,\e}} \otimes \sqrt{\Psi_{\theta,\e}} ) \VEC  K(\Theta)$ in the last line.

We conclude by applying the trace cyclic property, and extracting the minimum eigenvalue of $[\E_{\Theta}I_p(\theta) + J(\lambda)]^{-1}$.

\end{proof}
}\fi
\Cref{lem: application of van trees} may be interpreted according to the following intuition. To design a controller that attains low cost, it is essential to distinguish between two nearby instances of the underlying parameter, $\theta$ and $\theta'$, from the experimental data, $\calZ$. The Fisher Information term on the denominator of the bound in \Cref{lem: application of van trees} captures the ease with which we can distinguish between $\theta$ and an infinitesimally perturbed $\theta'$ from the collected data $\calZ$, and can be thought of as a signal-to-noise ratio. The derivative of the controller appearing on the numerator of the bound in \Cref{lem: application of van trees} is a change of variables term that accounts for the extent to which infinitesimal perturbations of the underlying parameter impact the optimal controller gain.  Sensitive perturbations are those which are difficult to detect from the collected data, yet lead to a large change in the controller gain. % The appearance of these quantities together illustrates the fact that the hardness of learning a LQR from a statistical point of view is dictated by infinitesimal perturbations of the underlying parameter.
Such perturbations dictate the statistical hardness of learning a LQR controller.
Motivated by this fact, we can select particularly sensitive perturbation directions of the underlying parameter which emphasize the hardness of the problem. To do so, we restrict the support of the prior $\lambda$ to a lower dimensional subspace. 
Before presenting this result, it will be useful to see the expression for Fisher information matrix from this experimental setup. It can be shown via the chain rule of Fisher Information that 
\ifnum\value{cdc}>0{
\begin{equation}
\begin{aligned}
I_p(\theta) &= \E_{\theta} \sum_{n=1}^N \sum_{t=0}^{T-1} \dop_\theta \VEC \bmat{A(\theta) &  B(\theta)}^\top \\ &\qquad \cdot [Z_{t,n} Z_{t,n}^\top \otimes \Sigma_W^{-1} ] \dop_\theta \VEC \bmat{A(\theta) & B(\theta)}, 
\end{aligned}
\end{equation}
}\else{
\begin{align}
I_p(\theta) = \E_{\theta} \sum_{n=1}^N \sum_{t=0}^{T-1} \dop_\theta \VEC \bmat{A(\theta) &  B(\theta)}^\top [Z_{t,n} Z_{t,n}^\top \otimes \Sigma_W^{-1} ] \dop_\theta \VEC \bmat{A(\theta) & B(\theta)}, 
\end{align}}\fi
where $Z_{t,n} = \bmat{X_{t,n} \\ U_{t,n}}$. See, for instance, Lemma 3.1 of \cite{ziemann2022regret}. With this in hand, the following Lemma provides a restriction to lower dimensional priors, which allows us to understand how poor conditioning of the information matrix along any particular parameter perturbation direction pushes through to a challenge in estimating the optimal controller.
%\Bruce{State this in the same form as above, but first state the above in a way that makes sense...}
\begin{lemma}
    \label{lem: coordinate transformation}
    Consider any matrix $V \in \R^{d_\Theta \times k}$ with $k \leq d_{\Theta}$ which has orthonormal columns. For any smooth prior $\lambda$ over $\curly{\theta + V \tilde \theta: \norm{\tilde \theta}\leq \e}$, and any $\pi \in \Pi^\mathsf{lin}$ with $\pi(X_t, \calZ) = \hat K(\calZ) X_t$,
    \ifnum\value{cdc}>0{
    \begin{align*}
        &\sup_{\theta'\in\calB(\theta,\e)}\mathsf{EC}_T^\pi(\theta')\geq \\& \frac{\tr\left(\Xi_{\theta,\e}  \E[\dop_\theta \VEC  K(\Theta) V \mathbf{1}_\calG] \E [\dop_\theta  \VEC K(\Theta) V 
         \mathbf{1}_\calG ]^\T \right)}{ \norm{V^\top \paren{\E \I_p(\Theta)+\J(\lambda)}V} },
    \end{align*}
    }
    \else{
    \begin{align*}
        \sup_{\theta'\in\calB(\theta,\e)}\mathsf{EC}_T^\pi(\theta')\geq\frac{\tr\left(\Xi_{\theta,\e}  \E[\dop_\theta \VEC  K(\Theta) V \mathbf{1}_\calG] \E [\dop_\theta  \VEC K(\Theta) V 
         \mathbf{1}_\calG ]^\T \right)}{ \norm{V^\top \paren{\E \I_p(\Theta)+\J(\lambda)}V} },
    \end{align*}}\fi
    where $\Xi_{\theta,\e}$ is defined in \eqref{eq: kronecker shorthand} and $\calG$ is defined in \Cref{lem: application of van trees}. 
    % If we also have $T \geq \frac{16 \norm{\Gamma_0}^2}{\lambda_{\min}(\Gamma_0)}$, then 
    % \begin{align*}
    %      \mathfrak{EC}^{\mathsf{lin}}_T(\theta,\e)&\geq \frac{\tr\left( \E[  (\Gamma_{\theta,\e} \otimes \Psi_{\theta,\e} )  \dop_\theta \VEC  K(\Theta) V \mathbf{1}_\calE] \E [\dop_\theta  \VEC K(\Theta) V \mathbf{1}_\calE]^\T \right)}{ \norm{V^\top \paren{\E \I_p(\Theta)+\J(\lambda)}V} }.
    % \end{align*}
\end{lemma}
\ifnum\value{cdc}>0{}\else{
\begin{proof}
    We may write $\Theta = \theta + V \tilde \Theta$, where  $\tilde \Theta \sim \tilde \lambda$, and $\tilde \lambda$ is a smooth prior on $\curly{\tilde\theta \in \R^k: \norm{\tilde \theta} \leq \epsilon}$. Defining $\tilde A(\tilde \theta) := A(\theta+V \tilde \theta)$, $\tilde B(\tilde \theta) := B(\theta+V\tilde \theta)$, and $\tilde K(\tilde \theta):=K(\theta+V \tilde \theta)$, we may instantiate the bound in \Cref{lem: application of van trees} over the lower dimensional parameter $\tilde \theta$. We have that the Jacobian of the contoller becomes
    $
        \dop_{\tilde \theta} \VEC \tilde K(\tilde \Theta) = \dop_{\theta} \VEC K(\Theta) V.
    $
    Similarly, the Jacobian arising in the Fisher information may be written $D_{\tilde \theta} \VEC \bmat{\tilde A(\tilde \Theta) & \tilde B(\tilde \Theta)} = D_{\theta} \VEC \bmat{A(\Theta) & B(\Theta)} V $. Lastly, the prior density of the lower dimensional parameter satisfies $\tilde J(\tilde \lambda) = V^\top J(\lambda) V$.  Then under this prior, the lower bound in \Cref{lem: application of van trees} becomes that in the lemma statement.
    % \begin{align*}
    %     \mathfrak{EC}^{\mathsf{lin}}_T(\theta,\e)&\geq \tr\left( \E[  (\Gamma_{\theta,\e} \otimes \Psi_{\theta,\e} )  \dop_\theta \VEC  K(\Theta) V \mathbf{1}_\calE] \left[  V^\top \paren{\E \I_p(\Theta)+\J(\lambda)}V \right]^{-1} \E [\dop_\theta  \VEC K(\Theta) V \mathbf{1}_\calE]^\T \right).
    % \end{align*}
    % Pulling out $\lambda_{\min}\paren{(V^\top \paren{\E \I_p(\Theta)+\J(\lambda)}V)^{-1}} = \frac{1}{\norm{V^\top \paren{\E \I_p(\Theta)+\J(\lambda)}V}}$ yields the desired result.  
\end{proof}}\fi

In the above lemma, the columns of $V$ may be interpreted as perturbation directions of the system parameters.

 % with $U_{t,n}$ generated according to $\pi$. 

% {\color{red} Introduce notation to make the following statment easier to parse. Write a short proof. For arxiv, state that the longer proof is in the appendix.}

We now upper bound the denominator arising in the above bound. In particular, we show how to bound the Fisher Information in any particular perturbation direction.
\begin{lemma}
    \label{lem: fisher bound}
    For any matrix $V\in \R^{d_{\Theta} \times k}$ with orthonormal columns, 
    \begin{align*}
        \norm{V^\top \E \brac{I_p(\Theta)} V} &\leq  TN \bar L, 
    \end{align*}
    where $\bar L =\sup_{\theta' \in \calB(\theta,\e)}   L(\theta')$ and %$L(\theta')$ is given as \Bruce{2 lines}
    \begin{equation}
    \begin{aligned}
        \label{eq: info bound L}
        % L&(\theta', w)\!=\!\frac{2 \nu_1(w) \nu_1 \bigg(\norm{\dlyap\paren{(A(\theta')+B(\theta')F)^\top, \Sigma_W}}}{\lambda_{\min}(\Sigma_W)} \\
        % & + \frac{2\nu_1(w) \beta \paren{\sum_{t=0}^{\infty} \norm{(A(\theta')  B(\theta')F)^t B}}^2 + 4 \nu_2(w) \beta }{\lambda_{\min}(\Sigma_W)}
        &L(\theta')=  \sup_{w \in \Span(V), \norm{w} \leq 1} \frac{4 }{\lambda_{\min}(\Sigma_W)} \\
        & \cdot 
         \Bigg(\nu_1(w) \bigg(\norm{\dlyap\paren{(A(\theta')+B(\theta')F)^\top, \Sigma_W}}\\
         &\!+\!  \sigma_{\tilde u}^2 \paren{\sum_{t=0}^{\infty} \norm{(A(\theta') \!+\!  B(\theta')F)^t B}}^2 \Bigg) \!+\! 2 \sigma_{\tilde u}^2 \nu_2(w)\bigg).
    \end{aligned}
    \end{equation}
    Here, $\nu_1(w) = \norm{D_{\theta} \VEC A(\theta') w}^2 +  2\norm{D_{\theta} \VEC B(\theta') w}^2 \norm{F}^2$ and $\nu_2(w) = \norm{D_{\theta} \VEC B(\theta') w}^2$ are change of coordinate terms that quantify the impact of the perturbation direction on the information upper bound. We recall that $ \sigma_{\tilde u}^2$ is the average exploratory input energy. 
        %\norm{  \dop_\theta \VEC \bmat{A(\theta)} v}^2 (\textrm{cont. from noise} + \textrm{cont. from input})  + 2 \beta TN \norm{  \dop_\theta \VEC \bmat{B(\theta)} v}^2  
\end{lemma}

\sloppy The quantity $\dlyap((A(\theta') + B(\theta') F)^\top, 
\Sigma_W)$, %\Tasos{Isn't it also the steady state state covariance under F?} 
in the above bound may be interpreted as either the steady state covariance during exploration in the absence of exploratory inputs, or the controllability gramian from the noise to the state. The quantity \ifnum\value{cdc}>0 {}\else{$\sum_{t=0}^\infty \norm{(A(\theta')+B(\theta') F)^t B(\theta')}$ bounds the $\calH_\infty$ norm of the closed-loop system during offline experimentation. Therefore, }\fi $ \sigma_{\tilde u}^2 \paren{\sum_{t=0}^\infty \norm{(A(\theta')+B(\theta') F)^t B(\theta')}}^2$ upper bounds the impact of exploratory input on the state during offline experimentation. 
%\Tasos{seems related to Hinfinity norm of offline feedback}.  
\ifnum\value{cdc}>0{}\else{
The proof of the above lemma applies repeated use of the triangle inequality, submultiplicativity, the Cauchy-Schwarz inequality. See \Cref{s: pf of fisher bound} for proof details.
}\fi

We now present our first main result: a non-asymptotic lower bound on the local minimax excess cost. As with \Cref{lem: application of van trees}, it is presented in two components: one that holds generally, and another that requires enough data such that any sufficiently good policy $\pi \in \Pi_{\mathsf{lin}}$ outputs a feedback controller $\hat K (\calZ)$ which
 is near optimal with high probability. Consequently, the burn-in times are larger for the second result, and the size of the prior, $\varepsilon$, is required to be small.  We drop the dependence of $A$, $B$, $P$, $\Psi$, $K$,  and $\Sigma_X$ on $\theta$ when the argument is clear from context. 

\begin{theorem}
    \label{thm: finite data bound}
    Consider any matrix $V \in \R^{d_\Theta \times k}$ with $k \leq d_{\Theta}$ which has orthonormal columns. Let
    % \ifnum\value{cdc}>0{
    \begin{align*}
        G &= \inf_{\theta^{'}, \tilde \theta \in \calB(\theta,\e)} \tr\bigg( \Xi_{\theta,\e} \dop_\theta \VEC  K(\theta')V \paren{\dop_\theta  \VEC K(\tilde \theta) V }^\T \bigg),
    \end{align*}
    % }\else{
    % \begin{align*}
    %     G &= \inf_{\theta', \tilde \theta \in \calB(\theta,\e)} \tr\left( (\Gamma_{\theta,\e} \otimes \Psi_{\theta,\e} )  \dop_\theta \VEC  K(\theta')V \paren{\dop_\theta  \VEC K(\tilde \theta) V }^\T \right),
    % \end{align*}
    % }\fi
    and $\bar L$ be as in \Cref{lem: fisher bound}. Also let $\Xi_{\theta,\e}$ be as defined in \eqref{eq: kronecker shorthand}. Then for any smooth prior $\lambda$ over $\curly{\theta + V \tilde \theta: \norm{\tilde \theta}\leq \e}$, 
    \begin{align}
        \label{eq: general LB after burn-in}
         \mathcal{EC}^{\mathsf{lin}}_T(\theta,\e)&\geq \frac{G}{8 NT \bar L}
    \end{align}
    is satisfied for
    \begin{blockquote}
        $1)$ $\Gamma_{\theta,\e}=\Sigma_W$ if $TN \geq \frac{ \norm{J(\lambda)}}{\bar L }$. \\
        $2)$ $\Gamma_{\theta,\e} = \Sigma_{\theta, \e}$ if  
         $T \geq \sup_{\theta' \in \calB(\theta,\e)} \frac{16 \norm{\Sigma_X(\theta')}^2}{\lambda_{\min}(\Sigma_X(\theta'))}$, $TN \geq \frac{1}{\bar L}{\max\curly{ \norm{J(\lambda)}, \frac{ G}{ \lambda_{\min}(\Sigma_W) \lambda_{\min}(R) \alpha^2}}}$, and  $\e \leq \min\curly{\frac{\alpha}{2c_1}, c_2}$, where 
            \begin{align*}
                    c_1 &= 84 \Phi^9  \tau(A_{cl}) \\
                    c_2 &= \frac{1}{ 10 \tau(A_{cl}) c_1 } \min\curly{(1+\norm{A_{cl}})^{-2}, (1 + \norm{P})^{-1}} \\
                     \Phi &= (1 + \max\curly{\norm{A}, \norm{B}, \norm{P}, \norm{K},\norm{R^{-1}}}) \\
                    \tau(A_{cl}) &= \paren{\sup_{k \geq 0} \curly{\norm{A_{cl}^k}\rho(A_{cl})^{-k}}}^2/(1-\rho(A_{cl})^2).
                \end{align*}
        % \end{itemize}
        % \end{blockquote2}
    \end{blockquote}

\end{theorem}
\ifnum\value{cdc}>0{

 The theorem follows by bounding the probability of the event $\calG$ in \Cref{lem: coordinate transformation}, then bounding the denominator with \Cref{lem: fisher bound}.
}\else{
%\Bruce{Modify the proof to attain a contradiction to the statement: There exists a policy $\pi \in \Pi^\mathsf{lin}$ such that the regret is less than $G/(8 NT \bar L)$. This change corrects for the change in presentation of Lemmas 2.1-2.3. }
\begin{proof}
    We must show that for all $\pi \in \Pi^{\mathsf{lin}}$, $\sup_{\theta'\in\calB(\theta,\e)} \mathsf{EC}_T^\pi(\theta') \geq \frac{G}{8 NT \bar L}$. Suppose that for some $\pi \in \Pi^\mathsf{lin}$, $\sup_{\theta' \in \calB(\theta, \e) }\mathsf{EC}_T^\pi(\theta') \leq \frac{G}{8 NT \bar L}$. We have by \Cref{lem: coordinate transformation} that
    \begin{align}
        \label{eq: single policy lb}
        \sup_{\theta'\in\calB(\theta,\e)} \mathsf{EC}_T^\pi(\theta') \geq \frac{\tr\left(\Xi_{\theta,\e}  \E[\dop_\theta \VEC  K(\Theta) V \mathbf{1}_\calG] \E [\dop_\theta  \VEC K(\Theta) V 
         \mathbf{1}_\calG ]^\T \right)}{ \norm{V^\top \paren{\E \I_p(\Theta)+\J(\lambda)}V} }.
    \end{align}
    The burn-in requirement $TN \geq \frac{\norm{J(\lambda)}}{\bar L}$ enables upper bounding $\norm{V^\top J(\lambda) V}$ by $TN \bar L$.  \Cref{lem: fisher bound} then allows us to upper bound the denominator in \eqref{eq: single policy lb} by $2TN \bar L$. 

    To remove the indicators from the lower bound, we take an infimum over $\tilde \theta, \theta' \in \calB(\theta,\e)$ to lower bound the numerator in \eqref{eq: single policy lb} by $\bfP[\calG]^2 G$. For case 1, we immediately have $\bfP[\calG]^2=  \bfP[\Omega]^2 = 1$. For case 2, we may leverage the assumptions that the prior is small and that the burn-in time is satisfied to show that $\bfP[\calG]^2 =\bfP[\calE]^2 \geq \frac{1}{4}$. See \Cref{s: proof of finite data bound}  for more details. This in turn implies that 
    \begin{align*}
        \sup_{\theta'\in\calB(\theta,\e)} \mathsf{EC}_T^\pi(\theta') \geq \frac{G}{8 TN\bar L}.
    \end{align*}
    Therefore, for all $\pi \in \Pi^{\mathsf{lin}}$, the above lower bound is satisfied. This implies that 
    \begin{align*}
        \mathcal{EC}^{\mathsf{lin}}_T(\theta,\e) = \inf_{\pi \in \Pi^\mathsf{lin}} \sup_{\theta' \in \calB(\theta,\e)} \mathsf{EC}_T^\pi(\theta') \geq \frac{G}{8 TN\bar L}.
    \end{align*}
    %\Tasos{it seems to me that while proving the lower bound on the probability event there is an interesting tension. Low regret leads to small K error leading to large regret. it might be worth to comment on that more} 
\end{proof}
}\fi
%The proof of the above result highlights an interesting tension. When the controller is near optimal as defined by the event $\calE$, the suboptimality is small.  However, in this setting, the lower bound on the state covariance of the learned controller is improved to $\Sigma_{\theta,\e}$ rather than $\Sigma_W$. This in turn increases the suboptimality lower bound. 
The above result holds non-asymptotically. It will be helpful to present the result asymptotically, as the number of experiments tends to $\infty$ for an understanding of the dependence on control-theoretic quantities. 
\begin{corollary}

\label{cor: asymptotic lower bound}
For any $\alpha \in (0, 1/2)$ and any matrix $V \in \R^{d_\Theta \times k}$ with $k \leq d_{\Theta}$ which has orthonormal columns, %, suppose $T \geq \frac{16 \norm{\Gamma_0}^2}{\lambda_{\min}(\Gamma_0)}$,
we have that
    \begin{align*}
        &\liminf_{N \to \infty} \sup_{\theta' \in \mathcal{B}(\theta, N^{-\alpha}) } N \mathsf{EC}_T^\pi(\theta')   \geq \frac{G}{8 T L(\theta)}, % \tr\left(   (\Gamma_0 \otimes (B^\top P B+R))  \dop_\theta \VEC  K(\theta) VV^\top \dop_\theta  \VEC    K(\theta)^\T \right).
        %\frac{\sigma_{\min}(B( \theta)^\top P(\theta) B(\theta) + R)  \sigma_{\min}(\Gamma_0(\theta))}{16 \mathcal{C}(T) N}  \norm{   \dop_\theta \VEC  K(\theta) V}_F^2
    \end{align*}
    holds always for $\Gamma=\Sigma_W$ and for $\Gamma=\frac{1}{2} \Sigma_X$ if $T \geq \frac{16 \norm{\Sigma_X}^2}{\lambda_{\min}(\Sigma_X)}$, where $L$ is as in \Cref{lem: fisher bound} and
    \begin{align*}
        G &= \tr\left( (\Gamma \otimes \Psi)  \dop_\theta \VEC  K(\theta )V \paren{\dop_\theta  \VEC K(\theta) V }^\T \right).
        %L &=  \sup_{w \in \Span(V), \norm{w} \leq 1}   \frac{8 }{\lambda_{\min}(\Sigma_W)}
         % \bigg(\paren{\norm{D_{\theta} \VEC A(\theta) w}^2 +  2\norm{D_{\theta} \VEC B(\theta) w}^2 \norm{F}^2} \\
         % &\cdot\bigg(\norm{\dlyap\paren{(A+BF)^\top, \Sigma_W}}
         % + \beta \sum_{t=0}^{\infty} \norm{\paren{(A +  BF)^t B}}^2 \bigg) + 2\beta \norm{D_{\theta} \VEC B(\theta) w}^2\bigg).
    \end{align*}
    
\end{corollary}
\ifnum\value{cdc}>0{}\else{\begin{proof}
    The burn-in requirements in \Cref{thm: finite data bound} are satisfied asymptotically, see \Cref{s: asymptotic lb proof} for more details. 
\end{proof}
}\fi

Using a similar argument to the derivations above, it can be shown that the global minimax complexity is infinite. 
\begin{corollary}
    \label{cor: global minimax}
    The global minimax excess cost is infinite for the class of scalar systems of the form: 
    \ifnum\value{cdc}>0{$X_{t+1} = a X_t + b U_t + W_t,$}
    \else{
    \[
        X_{t+1} = a X_t + b U_t + W_t,
    \]
    }\fi
    with $\theta = \bmat{a & b}^\top$, and $Q=R=\Sigma_W = \sigma_{\tilde u}^2= 1$. More precisely, for the class of stable scalar systems with the offline exploration policy $F=0$, we have
    \begin{align*}
        \liminf_{N \to \infty} \sup_{a, b: \abs{a} <1} N \mathsf{EC}_T^\pi(a,b) = \infty.
    \end{align*}
\end{corollary}
\ifnum\value{cdc}>0{}\else{
\begin{proof}
    We argue as in the proof of \Cref{cor: asymptotic lower bound} with $V = \bmat{0 & 1}^\top$. In this perturbation direction, the lower bound evaluates to $\frac{1}{T} \frac{b^2 P +1}{32} \frac{\partial K}{\partial b}(a,b)^2$. Considering $a = 1-\gamma$ and $b =\gamma$ for $0 < \gamma < 1$ and taking the limit as $\gamma\to 0$ results in the lower bound of $\infty$. For more details, see \Cref{s: global minimax pf}.
\end{proof}}\fi

%{\color{red} Write the version of the results that lower bounds the covariance by $\Sigma_W$ rather than $\Sigma_x$. There are a lot of choices for the presentation of this. I think it's best just to get some version of it down on the paper.}
\section{Consequences of the Lower Bound}
\label{s: consequences of lower bound}

In this section, we examine cases where the bound in \Cref{cor: asymptotic lower bound} has interpretable dependence upon system properties. To do so, we restrict attention to the setting where all system parameters are unknown, i.e. $\VEC \bmat{A(\theta) & B(\theta)} = \theta$. 
In this setting, the quantity $\dop_\theta \VEC \bmat{A(\theta) & B(\theta)}$ arising in the bounds from the previous section is the identity matrix. 

The derivative of the controller multiplied by a matrix with orthonormal columns,  $\dop_\theta \VEC K(\theta) V$, arises in the bounds from the previous section. In this section, this quantity is expressed in terms of the directional derivative of the controller in some direction $v$, denoted $d_v K(\theta)$. In particular, we represent the columns of $V$ as $v = \VEC \bmat{\Delta_A & \Delta_B}$ for arbitrary perturbations $\Delta_A$ of $A$ and $\Delta_B$ of $B$ which satisfy $\norm{\bmat{\Delta_A & \Delta_B}}_F=1$. The corresponding change in the closed-loop state matrix is denoted $\Delta_{A_{cl}} = \Delta_A + \Delta_B K$. 
Then the directional derivative of the controller is shown in Lemma B.1 of \cite{simchowitz2020naive} to be
    \begin{align}
        \label{eq: directional derivative of lqr}
        d_v K(\theta)\! =\! -\!\Psi^{-1} (\Delta_B^\top P A_{cl} \!+\! B^\top P \Delta_{A_{cl}} \!+\! B^\top P' A_{cl}),
    \end{align}
    where $P' = \texttt{dlyap}(A_{cl}, A_{cl}^\top P \Delta_{A_{cl}} + \Delta_{A_{cl}}^\top P A_{cl})$.
The subsequent sections study the bound from \Cref{cor: asymptotic lower bound} under various perturbations $\bmat{\Delta_A & \Delta_B}$.  \ifnum\value{cdc}>0{}\else{Proofs are deferred to \Cref{s: proofs of consequences}}\fi. 

%We show in Section~\ref{s: dimensional dependence} that the lower bound recovers the optimal dimensional dependence. In Section~\ref{s: exp}, we will demonstrate that there exist systems where the lower bound grows exponentially in the state dimesion. Lastly, in Section~\ref{s: system properties}, we will show that the bounds scale as expected with interpretable system properties such as the eigenvalues of the controllability gramian. 

% \Ingvar{Consider re-organizing so that system dependecies show up first. Our audience probably cares more about these!}

\subsection{Dimensional dependence}

In the setting of online LQR for an unknown system, recent works \citep{simchowitz2020naive, ziemann2022regret} obtaining lower bounds on the regret have %\footnote{Regret is the running average  suboptimality for online problems.} 
used perturbation directions which cause tension between identification and control \citep{polderman1986necessity}. In particular, they considered the set of perturbation directions 
\ifnum\value{cdc}>0{
\begin{equation}
\begin{aligned}
\label{eq: polderman perturbation}
&\boldsymbol \Delta \!=\!
\curly{\!\VEC\! \bmat{-\Delta K \!\! &\!\! \Delta } \!\bigg\vert \Delta \!\in\! \mathbb{R}^{\dx \times \du}, \! \norm{\bmat{-\Delta K \!\! &\!\! \Delta }}_F \!=\! 1\!}.
\end{aligned} 
\end{equation}
}\else{
\begin{align}
\label{eq: polderman perturbation}
\boldsymbol \Delta = \curly{\VEC \bmat{-\Delta K & \Delta } \bigg\vert \Delta \in \mathbb{R}^{\dx \times \du}, \norm{\bmat{-\Delta K & \Delta }}_F = 1}.
\end{align} }\fi
% Such perturbations are parameterized as $v(\Delta) \in \boldsymbol\Delta$. 
For all such perturbations, $\Delta_{A_{cl}} = 0$, making it impossible to distinguish between the true parameters and the perturbed parameters online without sufficient exploratory input noise. 

While the tension between identification and control is no longer present in the offline setting, this set of perturbation directions retains the benefit that the directional derivative in \eqref{eq: directional derivative of lqr} is easy to work with. In particular for any $v = \VEC \bmat{-\Delta K & \Delta} \in \boldsymbol \Delta$,
\begin{align}
    \label{eq: polderman derivative}
        d_v K(\theta) = -\Psi^{-1} \Delta^\top P A_{cl}. 
\end{align}
 
    As the matrices $\Delta$ parametrizing the set $\boldsymbol \Delta$ are $\dx \times \du$ dimensional, we may stack $\dx\du$ orthogonal vectors $v_i$ belonging $\boldsymbol \Delta$ into a matrix $V = \bmat{v_1 & \dots & v_{\dx \du}}$. This allows us to present a lower bound which demonstrates the dependence of the offline LQR problem upon the system dimensions $\dx$ and $\du$. 
    \begin{proposition}
    \label{prop: dimensional dependence}
    Suppose that $T \geq \frac{16 \norm{\Sigma_X}^2}{\lambda_{\min}(\Sigma_X)}$. Then for $\alpha \in (0,1/2)$, 
    \ifnum\value{cdc}>0{
    \begin{align*}
        \liminf_{N \to \infty} &\sup_{\theta' \in \mathcal{B}(\theta, N^{-\alpha}) } N \mathsf{EC}_T^\pi(\theta')   \geq \\  &\frac{\dx \du \lambda_{\min}(\Sigma_X - \Sigma_W) \lambda_{\min}(P)^2}{16 T \norm{\Psi} \norm{ \bmat{-K & I}}^2 \tilde L},
    \end{align*}
    }\else{
    \begin{align*}
        \liminf_{N \to \infty} \sup_{\theta' \in \mathcal{B}(\theta, N^{-\alpha}) } N \mathsf{EC}_T^\pi(\theta')   \geq  \frac{\dx \du \lambda_{\min}(\Sigma_X - \Sigma_W) \lambda_{\min}(P)^2}{16T \norm{\Psi} \norm{ \bmat{-K & I}}^2  \tilde L} 
    \end{align*}}\fi
    % where $\bar L$ is as in \Cref{prop: system theoretic}.
    where $\tilde L$ is given by $L(\theta)$ as in \eqref{eq: info bound L} by replacing $\nu_1$ with $1$ and $\nu_2$ with $1 + 2\norm{F}^2$.
    % \ifnum\value{cdc}>0{
    % \begin{align*}
    %      \bar L &= \frac{2 }{\lambda_{\min}(\Sigma_W)}
    %      \bigg( 2 \beta + \paren{1 +  2 \norm{F}^2}  \\ &\cdot \bigg(\norm{\dlyap\paren{(A(\theta)+B(\theta)F)^\top, \Sigma_W}} \\ &\qquad  
    %      + \beta \sum_{t=0}^{\infty} \norm{\paren{(A(\theta) +  B(\theta)F)^t B}}^2 \bigg) \bigg). 
    % \end{align*}
    % }\else{
    % \begin{align*}
    %      \bar L &= \frac{2 }{\lambda_{\min}(\Sigma_W)}
    %      \bigg( 2 \beta + \paren{1 +  2 \norm{F}^2}  \\ &\cdot \bigg(\norm{\dlyap\paren{(A(\theta)+B(\theta)F)^\top, \Sigma_W}} 
    %      + \beta \sum_{t=0}^{\infty} \norm{\paren{(A(\theta) +  B(\theta)F)^t B}}^2 \bigg) ). 
    % \end{align*}}\fi
    \end{proposition}
    
    In addition to the system dimensions, we can interpret the remaining system-theoretic parameters. 
    Note that $\tilde L$ bounds the information available from the offline experimentation. It depends on the norm of the controllability gramian from noise to the state, as well as $ \sigma_{\tilde u}^2 \paren{\sum_{t=0}^\infty \norm{(A+BF)^t B}}^2$, which bounds the impact of the exploratory input on the state. %consists of norms of quantities that resemble the Gramian of the closed-loop system under the exploration policy from the noise, and from the exploratory input. The energy of the exploratory input signal also makes an appearance. As the exploratory budget becomes small, $\beta \to 0$, the $\bar L$ approaches the norm of the controllabiltiy gramian from noise of the closed-loop system, scaled by the process noise level and the norm of the pre-stabilizing controller. 
    % As in \Cref{prop: dimensional dependence}, $\lambda_{\min}(\Sigma_X - \Sigma_W)$ highlights the dependence on the closed-loop state covariance, and $\bar L$ describes the impact of the controllability of the closed-loop system under the pre-stabilizing controller, as well as the input budget.
     %In particular, the lower bound becomes small as the system becomes easily controllable during experimentation. Similarly, the lower bound becomes small when the exploratory input budget is large. 
     The $\Psi$ in the denominator of the above bound may scale as $\lambda_{\max}(P)$, and therefore effectively cancels a $\lambda_{\min}(P)$ in the numerator for well-conditioned problems. This leaves a single $\lambda_{\min}(P)$ in the numerator. As $x^\top P x$ is the optimal objective value of the noiseless LQR problem starting from initial state $x$, the appearance of $\lambda_{\min}(P)$ in the bound captures the fact that as the system becomes harder to control, it also becomes harder to learn to control. 
     Lastly, the variance term $\lambda_{\min}(\Sigma_X  - \Sigma_W)$ implies that the excess cost is large when the optimal closed-loop system has a large state covariance relative to the process noise covariance. 

    \ifnum\value{cdc}>0{
    \begin{remark}
        The dimensional dependence $\dx\du$  is optimal up to constant factors for classes of under-actuated systems in which remaining system-theoretic quantities are constant with respect to system dimension. In particular, suppose that  $\tilde U_{t,n} \sim \calN(0, \sigma_{\tilde u}^2 I)$, and that $N \geq c \dx$, for some universal constant $c$. In this case, Theorem 2 of \cite{mania2019certainty} combined with Theorem 5.4 in \cite{tu2022learning} demonstrate 
        the upper bound on the excess cost scales with  $\frac{\dx \du}{NT}$. 
    \end{remark}
    }\else{
    \begin{remark}
        The dimensional dependence $\dx\du$ in the above bound is optimal up to constant factors when $\du \leq \dx$. To see that this is so, observe that Theorem 2 of \cite{mania2019certainty} demonstrates an upper bound on the excess cost that scales as $\du \e^2$, where $\e^2$ bounds the system identification error, $\max\curly{\norm{\hat A - A}^2, \norm{\hat B - B}^2}$. A consequence of Theorem 5.4 in \cite{tu2022learning} is that if we apply exploratory inputs which are generated from a Gaussian distribution with mean zero and covariance $ \sigma_{\tilde u}^2 I$, then the upper bound on the system identification error scales as $\frac{\dx + \du}{NT}$. In particular, as long as number of offline trajectories $N$ exceeds $c \dx$, for some universal constant $c$, then 
        $\max\curly{\norm{\hat A - A}^2, \norm{\hat B - B}^2} \lesssim \norm{\Sigma_W} \frac{\dx+\du}{NT \lambda_{\min}(\textrm{controllability gramian})}$. 
        Consequently, the upper bound on the excess cost scales with $\frac{\du (\dx + \du)}{NT} \lesssim \frac{\dx \du}{NT}$ in the underactuated setting.
        %\Tasos{typo? $\du (\dx + \du)$?}. 
        Therefore, for classes of systems where the remaining system-theoretic quantities are constant with respect to system dimension, the bound is optimal in the dimension. 
    \end{remark}
    }\fi

\label{s: dimensional dependence}
\subsection{Exponential Lower Bounds}
\label{s: exp}
% \Tasos{I still have some questions regarding Lemma 2.4. Shouldn't we have a product between $D_{\theta}B(\theta)$ and the covariance of $X$ even for $F=0$? For example in the system below. We have $x_{t,1}=2^{n-1}(1+\theta)u_{t-n}+noise$, if we take $1+\theta$ to be the last coordinate of $B$. It seems to me that the  Fisher information should be exponential (or $\nu_1\neq 0$), not equal to $16\beta$. With a small perturbation of $B$ we get a huge effect on $x_{t,1}$. Am I missing something? Is it hidden under the Sigma noise somehow?}
The previous section demonstrated a lower bound that scales linearly with $\dx \du$. Prior work \citep{tsiamis2022learning} has shown that in the setting of online LQR, there exist classes of systems where the lower bounds on the regret may scale exponentially with the state dimension. This is shown by demonstrating that particular system-theoretic terms, which are often treated as constant with respect to dimension, may actually grow exponentially with the state dimension. 
We demonstrate that in the setting of offline LQR, such systems still cause exponential dependence on dimension. Furthermore, because there are fewer restrictions upon the perturbation directions in the lower bound for the offline setting, we construct a simpler class of a systems which exhibits this behavior. In particular, consider the system
\ifnum\value{cdc}>0{
\begin{equation}
\begin{aligned}
    \label{eq: integrator}
    A &= \bmat{ \rho & 2 & 0 & & 0 & 0 \\ %0 & \rho & 2 & & 0 & 0 \\ 
    & & & \ddots  \\ 0 & 0 & 0 & & \rho & 2 \\ 0 & 0 & 0 & & 0 & \rho}, \quad B = \bmat{0 \\ \vdots \\ 0 \\ 1}, %\\ 
    %F&=0, Q = I, R = 1, \Sigma_W = I,
\end{aligned}
\end{equation}
}\else{
\begin{align}
    \label{eq: integrator}
    A = \bmat{ \rho & 2 & 0 & & 0 & 0 \\ 0 & \rho & 2 & & 0 & 0 \\ & & & \ddots  \\ 0 & 0 & 0 & & \rho & 2 \\ 0 & 0 & 0 & & 0 & \rho}, \quad B = \bmat{0 \\ 0\\ \vdots \\ 0 \\ 1}, %\quad F=0, Q = I, R = 1, \Sigma_W = I,
\end{align}}\fi
with $0 < \rho < 1$, $F=0$, $Q = I$, $R = 1$, and $\Sigma_W = I$.
\ifnum\value{cdc}>0{
Using the perturbation direction $V =  \VEC \bmat{0 & B/\norm{B}_F}$, \Cref{cor: asymptotic lower bound} may be used to reach the following conclusion.
}
\else{
Let $V = \VEC \bmat{0 & B/\norm{B}_F}$. Then the quantity $L(\theta)$ in \Cref{cor: asymptotic lower bound} becomes $8  \sigma_{\tilde u}^2$, as $\nu_1(V) = 0$. Meanwhile, (using the option $\Gamma= \Sigma_W = I$), the quantity $G$ becomes
% \ifnum\value{cdc}>0{
% \begin{equation}
% \begin{aligned}
%     \label{eq: F LB exp}
%     %&\tr\left(   (I \otimes \Psi)  \dop_\theta \VEC  K(\theta) VV^\top \dop_\theta  \VEC    K(\theta)^\T \right) \\
%     %& \qquad 
%     G = \tr(\Psi d_V K(\theta) d_V K(\theta)^\top).
% \end{aligned}
% \end{equation}
% }\else{
\begin{equation}
\begin{aligned}
    \label{eq: F LB exp}
    G = \tr\left(   (I \otimes \Psi)  \dop_\theta \VEC  K(\theta) VV^\top \dop_\theta  \VEC    K(\theta)^\T \right) = \tr(\Psi d_V K(\theta) d_V K(\theta)^\top).
\end{aligned}
\end{equation}
% }\fi
Using this insight, we may show that the lower bound grows exponentially with the system dimension. 
}\fi
\begin{proposition}
    \label{prop: exponential}
    For the system in \eqref{eq: integrator} suppose $\dx \geq 3$. Then for $\alpha \in (0,1/2)$,
    \begin{align*}
        \liminf_{N \to \infty} \sup_{\theta' \in \mathcal{B}(\theta, N^{-\alpha}) } N \mathsf{EC}_T^\pi(\theta')   \geq \frac{\rho^2}{256 T  \sigma_{\tilde u}^2} 4^{\dx-2}.
    \end{align*}
\end{proposition}

We have therefore demonstrated that accurately learning the LQR controller from offline data may require an amount of data that is exponential in the state dimension. The reason that this system is particularly challenging to learn to control is that a small misidentification of $B$ causes the learner to apply slightly suptoptimal control inputs, which are then amplified by the off-diagonal terms of $A$.  The construction used, \eqref{eq: integrator}, avoids the two subsystem example that was used to derive exponential lower bounds for online LQR in \cite{tsiamis2022learning}. A crucial reason that we are able to bypass such a construction in the offline setting is that the dominant statistical rate of $\frac{1}{NT}$ for offline LQR is present for any perturbation direction of the underlying parameters. In contrast, the regret in the online setting only has the dominant statistical rate in the directions defined by the perturbation set in \eqref{eq: polderman perturbation}.

% \begin{remark}
%     {\color{red} It appears as though it should also be possible to get exponential lower bounds from this system in the online setting. Currently I have reduced the lower bound to $\tr(K\Gamma_0 K^\top)/KK^\top$. Leave working out this remark for the journal submission.} 
% \end{remark}

 \subsection{Interesting System-Theoretic Quantities}
\label{s: system-theoretic}

% \Bruce{I like some parts of the following discussion, but I'm not sure it's all necessary/the best location for it.}
% \Cref{s: system-theoretic} focuses on on a perturbation direction for which the lower bound aligns with intuition about when it should be difficult to learn a controller. 
%\Cref{s: dimensional dependence} demonstrates a set of perturbation directions for which the dependence on the dimension is optimal. \Cref{s: exp} illustrates that for reasonable systems, the system theoretic quantities appearing in the lower bound may grow exponentially with the system dimension. 
A consequence of the result in \Cref{s: exp} is that treating system-theoretic quantities as constant with respect to dimension, as is done in \Cref{s: dimensional dependence}, may fail to capture the difficulty of the problem. This leads to unfavorable aspects of the lower bound in \Cref{s: dimensional dependence}, such as the dependence of the denominator on $\norm{K}$. Such an appearance indicates that for systems where the optimal LQR has a large gain, the lower bound becomes small. This is in contrast to our expectations, as a large optimal gain is often indicative of poor controllability (consider a scalar system, with $B\to 0$). %As the control problem becomes more difficult in such a setting, intuition tells us that the suboptimality of a learned controller  should become large.  
% By focusing on system-theoretic quantities, the result in \Cref{s: system-theoretic} avoids several of the pitfalls of that in \Cref{s: dimensional dependence}. 

Motivated by the above discussion, we focus our attention on deriving bounds which have favorable dependence upon system-theoretic quantities. To do so, we examine a perturbation direction for which the lower bound from \Cref{cor: asymptotic lower bound} reduces to easily interpretable quantities which align with our intuition. By taking $V=  \VEC \frac{\bmat{A & B}}{\norm{\bmat{A & B}}_F}$, the directional derivative expression from \eqref{eq: directional derivative of lqr} reduces to 
\begin{align}
    \label{eq: system properties directional derivative}
    d_V  K(\theta) = \frac{2\Psi^{-1} (B^\top \texttt{dlyap}(A_{cl}, P) A_{cl})}{\norm{\bmat{A & B}}_F} .
\end{align}
\ifnum\value{cdc}>0{}\else{
Then the quantity $G$ in \Cref{cor: asymptotic lower bound} (using $\Gamma = \Sigma_X$), 
is
\begin{equation}
\label{eq: G system theoretic}
\begin{aligned}
    &\tr\paren{(\Sigma_X \otimes (B^\top P B + R))\dop_\theta \VEC  K(\theta )V \paren{\dop_\theta  \VEC K(\theta) V }^\T } \\
    &=\tr((B^\top P B+R) d_V K(\theta) \Sigma_X d_V K(\theta)^\top) \\
    &=\frac{4}{\norm{\bmat{A & B}}_F^2} \tr((B^\top P B+R)^{-1} B^\top \dlyap\paren{A_{cl}, P} A_{cl} \Sigma_X A_{cl}^\top \dlyap\paren{A_{cl}, P} B).
    %&\geq \frac{4 \lambda_{\min}(\Gamma_0 - \Sigma_W)}{\norm{\bmat{A & B}}_F^2 } \frac{\tr\paren{B^\top P B}}{\norm{B^\top P B+ R}} \lambda_{\min}(\dlyap(A_{cl}, P))
\end{aligned}
\end{equation}}\fi
This leads to the following proposition.
\begin{proposition}
    \label{prop: system theoretic}
    Suppose that $R$ and $B^\top P B$ are simultaneously diagonalizable by $U$: $B^\top P B = U \Lambda_{B^\top P B} U^\top$ and $R = U \Lambda_R U^\top$, where $\Lambda_{B^\top P B}$ and $\Lambda_R$ are diagonal. Also suppose that the diagonal entries of $\Lambda_{B^\top P B}$ are sorted in non-ascending order.  Assume $T \geq \frac{16 \norm{\Sigma_X}^2}{\lambda_{\min}(\Sigma_X)}$. Let $\tilde L$ be as in \Cref{prop: dimensional dependence}. Then for $\alpha \in (0,1/2)$
    \ifnum\value{cdc}>0{
    \begin{align*}
        &\liminf_{N \to \infty} \sup_{\theta' \in \mathcal{B}(\theta, N^{-\alpha}) } N \mathsf{EC}_T^\pi(\theta')  \geq  \frac{ \lambda_{\min}(\Sigma_X - \Sigma_W)}{2T \norm{\bmat{A & B}}_F^2 \bar L} 
        \\ &\qquad 
        \cdot\inf_{i \in [\du]} \frac{\lambda_{i}(B^\top P B)}{\lambda_i(B^\top P B) + \Lambda_{R, ii}} \sum_{j=1}^{\du} \lambda_{n-j} (\dlyap(A_{cl}, P)).
    \end{align*}
    }\else{
    \begin{align*}
        &\liminf_{N \to \infty} \sup_{\theta' \in \mathcal{B}(\theta, N^{-\alpha}) } N \mathsf{EC}_T^\pi(\theta')  \geq  \frac{ \lambda_{\min}(\Sigma_X - \Sigma_W)}{2T \norm{\bmat{A & B}}_F^2 \tilde L} 
        % \\ &\quqad \cdot
        \inf_{i \in [\du]} \frac{\lambda_{i}(B^\top P B)}{\lambda_i(B^\top P B) + \Lambda_{R, ii}} \sum_{j=1}^{\du} \lambda_{n-j} (\dlyap(A_{cl}, P)).
    \end{align*}}\fi
    %where $\tilde L$ is as in \Cref{prop: dimensional dependence}.  
    % \begin{align*}
    %      \bar L &= \frac{8 }{\lambda_{\min}(\Sigma_W)}
    %      \bigg(\paren{1 +  2 \norm{F}^2} \bigg(\norm{\dlyap\paren{(A(\theta)+B(\theta)F)^\top, \Sigma_W}} \\
    %      &\qquad
    %      + \beta \sum_{t=0}^{\infty} \norm{\paren{(A(\theta) +  B(\theta)F)^t B}}^2 \bigg) + 2\beta \bigg). 
    % \end{align*}
\end{proposition}

\ifnum\value{cdc}>0{The assumption that $R$ and $B^\top P B$ are simultaneously diagonalizable is satisfied if $R$ is chosen as a scalar multiple of the identity. }\else{
As $R$ is often chosen to be a scalar multiple of the identity for LQR problems, the assumption that $R$ and $B^\top P B$ are simultaneously diagonalizable is often satisfied. If we additionally have $R \preceq B^\top P B$, then $\inf_{i \in [\du]} \frac{\lambda_i(B^\top P B)}{\lambda_i(B^\top P B) + \Lambda_{R, ii}} \geq \frac{1}{2}$. }\fi
As in \Cref{prop: dimensional dependence}, $\lambda_{\min}(\Sigma_X - \Sigma_W)$ highlights the dependence on the closed-loop state covariance, 
\ifnum\value{cdc}>0{ and $\tilde L$ describes the how easily system is excited offline. }\else{
and $\tilde L$ describes the impact of the controllability of the closed-loop system under the pre-stabilizing controller, as well as the input budget.}\fi \ifnum\value{cdc}>0{}\else{ Note that $\tilde L$ provides an upper bound on the information in the face of an optimal offline exploration policy. Studying it may therefore assist with experiment design, as in \cite{wagenmaker2021task}.}\fi\  Rather than the appearance of $\norm{\bmat{-K & I}}$ on the denominator, as we saw in \Cref{prop: dimensional dependence},  we have $\norm{\bmat{A & B}}_F$. Therefore, the bound does not diminish as a result of a large optimal controller gain. 
Lastly, observe that $ \sum_{j=1}^{du} \lambda_{n-j} (\dlyap(A_{cl}, P))$ replaces $\lambda_{\min}(P)$ from \Cref{prop: dimensional dependence}. %The new quantity is strictly larger. 
This quantity captures the $\du$ smallest eigenvalues rather than just the smallest. %, as appeared in \Cref{prop: dimensional dependence}. 
If $\du = \dx$, we get all eigenvalues of $\dlyap(A_{cl}, P)$. Further note that the eigenvalues of $\dlyap(A_{cl}, P)$ diverge as $A_{cl}$ approaches marginal stability, leading to an infinite excess cost.

\section{Conclusion}
\label{s: conclusion}

We presented lower bounds for offline linear-quadratic control problems. The focus was to understand the fundamental limitations of learning controllers from offline data in terms of system-theoretic properties. 
 Several interesting consequences arose, such as the fact that our lower bound achieves the optimal dimensional dependence $\dx \du$ for underactuated systems. We also showed that there exist classes of systems where the sample complexity is exponential with the system dimension, $\dx$. We finally demonstrated that the lower bound scales in a natural way with familiar system-theoretic constants including the eigenvalues of the Riccati solution. An avenue for future work is extension of the lower bounds to the partially observed setting.  %to formulate the exploration assumptions to capture the problem of transfer learning by enforcing that the exploration policy behaves near optimally in some objective. 

\ifnum\value{cdc}>0{}\else{\section*{Acknowledgements} 
% We thank Thomas Zhang for helpful discussions. 
Bruce D. Lee is supported by the DoD through a National Defense Science \& Engineering Fellowship. Ingvar Ziemann is supported by a Swedish Research Council International Postdoc grant.  Henrik Sandberg is supported by the Swedish Research Council (grant 2016-00861). Nikolai Matni is partially supported by NSF CAREER award ECCS-2045834.}\fi

\ifnum\value{cdc}>0{
\bibliographystyle{IEEEtran}}\else{
\bibliographystyle{abbrvnat}}\fi
\bibliography{refs}

\begin{thebibliography}{45}
\providecommand{\natexlab}[1]{#1}
\providecommand{\url}[1]{\texttt{#1}}
\expandafter\ifx\csname urlstyle\endcsname\relax
  \providecommand{\doi}[1]{doi: #1}\else
  \providecommand{\doi}{doi: \begingroup \urlstyle{rm}\Url}\fi

\bibitem[Abbasi-Yadkori and Szepesv{\'a}ri(2011)]{abbasi2011regret}
Y.~Abbasi-Yadkori and C.~Szepesv{\'a}ri.
\newblock {Regret Bounds for the Adaptive Control of Linear Quadratic Systems}.
\newblock In \emph{Proceedings of the 24th Annual Conference on Learning
  Theory}, pages 1--26, 2011.

\bibitem[Abeille and Lazaric(2018)]{abeille2018improved}
M.~Abeille and A.~Lazaric.
\newblock {Improved Regret Bounds for Thompson Sampling in Linear Quadratic
  Control Problems}.
\newblock \emph{Proceedings of Machine Learning Research}, 80, 2018.

\bibitem[{\AA}str{\"o}m and Wittenmark(1973)]{aastrom1973self}
K.~J. {\AA}str{\"o}m and B.~Wittenmark.
\newblock On self tuning regulators.
\newblock \emph{Automatica}, 9\penalty0 (2):\penalty0 185--199, 1973.

\bibitem[Azar et~al.(2017)Azar, Osband, and Munos]{azar2017minimax}
M.~G. Azar, I.~Osband, and R.~Munos.
\newblock Minimax regret bounds for reinforcement learning.
\newblock In \emph{International Conference on Machine Learning}, pages
  263--272. PMLR, 2017.

\bibitem[Bobrovsky et~al.(1987)Bobrovsky, Mayer-Wolf, and
  Zakai]{bobrovsky1987some}
B.-Z. Bobrovsky, E.~Mayer-Wolf, and M.~Zakai.
\newblock {Some Classes of Global Cram{\'e}r-Rao Bounds}.
\newblock \emph{The Annals of Statistics}, pages 1421--1438, 1987.

\bibitem[Cassel et~al.(2020)Cassel, Cohen, and Koren]{cassel2020logarithmic}
A.~Cassel, A.~Cohen, and T.~Koren.
\newblock {Logarithmic Regret for Learning Linear Quadratic Regulators
  Efficiently}.
\newblock \emph{arXiv preprint arXiv:2002.08095}, 2020.

\bibitem[Chen and Nett(1993)]{chen1993caratheodory}
J.~Chen and C.~N. Nett.
\newblock The caratheodory-fejer problem and {$\mathcal{H}_\infty$} a time
  domain approach.
\newblock In \emph{Proceedings of 32nd IEEE Conference on Decision and
  Control}, pages 68--73. IEEE, 1993.

\bibitem[Cohen et~al.(2019)Cohen, Koren, and Mansour]{cohen2019learning}
A.~Cohen, T.~Koren, and Y.~Mansour.
\newblock {Learning Linear-Quadratic Regulators Efficiently with only {$\sqrt
  {T}$} Regret}.
\newblock \emph{arXiv preprint arXiv:1902.06223}, 2019.

\bibitem[Dean et~al.(2018)Dean, Mania, Matni, Recht, and Tu]{dean2018regret}
S.~Dean, H.~Mania, N.~Matni, B.~Recht, and S.~Tu.
\newblock {Regret Bounds for Robust Adaptive Control of the Linear Quadratic
  Regulator}.
\newblock In \emph{Advances in Neural Information Processing Systems}, pages
  4188--4197, 2018.

\bibitem[Dean et~al.(2019)Dean, Mania, Matni, Recht, and Tu]{dean2019sample}
S.~Dean, H.~Mania, N.~Matni, B.~Recht, and S.~Tu.
\newblock {On the Sample Complexity of the Linear Quadratic Regulator}.
\newblock \emph{Foundations of Computational Mathematics}, pages 1--47, 2019.

\bibitem[Faradonbeh et~al.(2018)Faradonbeh, Tewari, and
  Michailidis]{faradonbeh2018finite}
M.~K.~S. Faradonbeh, A.~Tewari, and G.~Michailidis.
\newblock Finite time identification in unstable linear systems.
\newblock \emph{Automatica}, 96:\penalty0 342--353, 2018.

\bibitem[Faradonbeh et~al.(2020)Faradonbeh, Tewari, and
  Michailidis]{faradonbeh2020input}
M.~K.~S. Faradonbeh, A.~Tewari, and G.~Michailidis.
\newblock {Input Perturbations for Adaptive Control and Learning}.
\newblock \emph{Automatica}, 117:\penalty0 108950, 2020.

\bibitem[Fazel et~al.(2018)Fazel, Ge, Kakade, and Mesbahi]{fazel2018global}
M.~Fazel, R.~Ge, S.~Kakade, and M.~Mesbahi.
\newblock Global convergence of policy gradient methods for the linear
  quadratic regulator.
\newblock In \emph{International Conference on Machine Learning}, pages
  1467--1476. PMLR, 2018.

\bibitem[Helmicki et~al.(1991)Helmicki, Jacobson, and
  Nett]{helmicki1991control}
A.~J. Helmicki, C.~A. Jacobson, and C.~N. Nett.
\newblock Control oriented system identification: a worst-case/deterministic
  approach in {$\mathcal{H}_\infty$}.
\newblock \emph{IEEE Transactions on Automatic control}, 36\penalty0
  (10):\penalty0 1163--1176, 1991.

\bibitem[Ibragimov and Has'minskii(2013)]{ibragimov2013statistical}
I.~A. Ibragimov and R.~Z. Has'minskii.
\newblock \emph{{Statistical Estimation: Asymptotic Theory}}, volume~16.
\newblock Springer Science {\&} Business Media, 2013.

\bibitem[Jedra and Proutiere(2019)]{jedra2019sample}
Y.~Jedra and A.~Proutiere.
\newblock {Sample Complexity Lower Bounds for Linear System Identification}.
\newblock In \emph{2019 IEEE 58th Conference on Decision and Control (CDC)},
  pages 2676--2681. IEEE, 2019.

\bibitem[Jedra and Proutiere(2021)]{jedra2021minimal}
Y.~Jedra and A.~Proutiere.
\newblock Minimal expected regret in linear quadratic control.
\newblock \emph{arXiv preprint arXiv:2109.14429}, 2021.

\bibitem[Lai(1986)]{lai1986asymptotically}
T.~L. Lai.
\newblock {Asymptotically Efficient Adaptive Control in Stochastic Regression
  Models}.
\newblock \emph{Advances in Applied Mathematics}, 7\penalty0 (1):\penalty0
  23--45, 1986.

\bibitem[Lai and Wei(1986)]{lai1986extended}
T.~L. Lai and C.-Z. Wei.
\newblock {Extended Least squares and their Applications to Adaptive Control
  and Prediction in Linear Systems}.
\newblock \emph{IEEE Transactions on Automatic Control}, 31\penalty0
  (10):\penalty0 898--906, 1986.

\bibitem[Lee and Lamperski(2020)]{lee2020non}
B.~Lee and A.~Lamperski.
\newblock Non-asymptotic closed-loop system identification using autoregressive
  processes and hankel model reduction.
\newblock In \emph{2020 59th IEEE Conference on Decision and Control (CDC)},
  pages 3419--3424. IEEE, 2020.

\bibitem[Levine et~al.(2016)Levine, Finn, Darrell, and Abbeel]{levine2016end}
S.~Levine, C.~Finn, T.~Darrell, and P.~Abbeel.
\newblock End-to-end training of deep visuomotor policies.
\newblock \emph{The Journal of Machine Learning Research}, 17\penalty0
  (1):\penalty0 1334--1373, 2016.

\bibitem[Levine et~al.(2020)Levine, Kumar, Tucker, and Fu]{levine2020offline}
S.~Levine, A.~Kumar, G.~Tucker, and J.~Fu.
\newblock Offline reinforcement learning: Tutorial, review, and perspectives on
  open problems.
\newblock \emph{arXiv preprint arXiv:2005.01643}, 2020.

\bibitem[Ljung(1998)]{ljung1998system}
L.~Ljung.
\newblock \emph{System identification}.
\newblock Springer, 1998.

\bibitem[Mania et~al.(2019)Mania, Tu, and Recht]{mania2019certainty}
H.~Mania, S.~Tu, and B.~Recht.
\newblock {Certainty Equivalence is Efficient for Linear Quadratic Control}.
\newblock In \emph{Advances in Neural Information Processing Systems}, pages
  10154--10164, 2019.

\bibitem[Ouyang et~al.(2017)Ouyang, Gagrani, and Jain]{ouyang2017control}
Y.~Ouyang, M.~Gagrani, and R.~Jain.
\newblock {Control of Unknown Linear Systems with Thompson Sampling}.
\newblock In \emph{2017 55th Annual Allerton Conference on Communication,
  Control, and Computing (Allerton)}, pages 1198--1205. IEEE, 2017.

\bibitem[Oymak and Ozay(2019)]{oymak2019non}
S.~Oymak and N.~Ozay.
\newblock Non-asymptotic identification of lti systems from a single
  trajectory.
\newblock In \emph{2019 American control conference (ACC)}, pages 5655--5661.
  IEEE, 2019.

\bibitem[Polderman(1986)]{polderman1986necessity}
J.~W. Polderman.
\newblock {On the Necessity of Identifying the True Parameter in Adaptive LQ
  Control}.
\newblock \emph{Systems \& control letters}, 8\penalty0 (2):\penalty0 87--91,
  1986.

\bibitem[Sarkar and Rakhlin(2019)]{sarkar2019near}
T.~Sarkar and A.~Rakhlin.
\newblock {Near Optimal Finite Time Identification of Arbitrary Linear
  Dynamical Systems}.
\newblock In \emph{International Conference on Machine Learning}, pages
  5610--5618, 2019.

\bibitem[Sarkar et~al.(2021)Sarkar, Rakhlin, and Dahleh]{sarkar2021finite}
T.~Sarkar, A.~Rakhlin, and M.~A. Dahleh.
\newblock Finite time lti system identification.
\newblock \emph{The Journal of Machine Learning Research}, 22\penalty0
  (1):\penalty0 1186--1246, 2021.

\bibitem[Silver et~al.(2017)Silver, Schrittwieser, Simonyan, Antonoglou, Huang,
  Guez, Hubert, Baker, Lai, Bolton, et~al.]{silver2017mastering}
D.~Silver, J.~Schrittwieser, K.~Simonyan, I.~Antonoglou, A.~Huang, A.~Guez,
  T.~Hubert, L.~Baker, M.~Lai, A.~Bolton, et~al.
\newblock Mastering the game of go without human knowledge.
\newblock \emph{nature}, 550\penalty0 (7676):\penalty0 354--359, 2017.

\bibitem[Simchowitz and Foster(2020)]{simchowitz2020naive}
M.~Simchowitz and D.~Foster.
\newblock Naive exploration is optimal for online lqr.
\newblock In \emph{International Conference on Machine Learning}, pages
  8937--8948. PMLR, 2020.

\bibitem[Simchowitz et~al.(2018)Simchowitz, Mania, Tu, Jordan, and
  Recht]{simchowitz2018learning}
M.~Simchowitz, H.~Mania, S.~Tu, M.~I. Jordan, and B.~Recht.
\newblock Learning without mixing: Towards a sharp analysis of linear system
  identification.
\newblock In \emph{Conference On Learning Theory}, pages 439--473. PMLR, 2018.

\bibitem[Simon(1956)]{simon1956dynamic}
H.~A. Simon.
\newblock {Dynamic Programming under Uncertainty with a Quadratic Criterion
  Function}.
\newblock \emph{Econometrica, Journal of the Econometric Society}, pages
  74--81, 1956.

\bibitem[S{\"o}derstr{\"o}m(2002)]{soderstrom2002discrete}
T.~S{\"o}derstr{\"o}m.
\newblock \emph{{Discrete-Time Stochastic systems: Estimation and Control}}.
\newblock Springer Science \& Business Media, 2002.

\bibitem[Tsiamis and Pappas(2019)]{tsiamis2019finite}
A.~Tsiamis and G.~J. Pappas.
\newblock Finite sample analysis of stochastic system identification.
\newblock In \emph{2019 IEEE 58th Conference on Decision and Control (CDC)},
  pages 3648--3654. IEEE, 2019.

\bibitem[Tsiamis and Pappas(2021)]{tsiamis2021linear}
A.~Tsiamis and G.~J. Pappas.
\newblock Linear systems can be hard to learn.
\newblock \emph{arXiv preprint arXiv:2104.01120}, 2021.

\bibitem[Tsiamis et~al.(2022{\natexlab{a}})Tsiamis, Ziemann, Matni, and
  Pappas]{tsiamis2022statistical}
A.~Tsiamis, I.~Ziemann, N.~Matni, and G.~J. Pappas.
\newblock Statistical learning theory for control: A finite sample perspective.
\newblock \emph{arXiv preprint arXiv:2209.05423}, 2022{\natexlab{a}}.

\bibitem[Tsiamis et~al.(2022{\natexlab{b}})Tsiamis, Ziemann, Morari, Matni, and
  Pappas]{tsiamis2022learning}
A.~Tsiamis, I.~Ziemann, M.~Morari, N.~Matni, and G.~J. Pappas.
\newblock Learning to control linear systems can be hard.
\newblock In \emph{Conference on Learning Theory}, pages 3820--3857. PMLR,
  2022{\natexlab{b}}.

\bibitem[Tu et~al.(2022)Tu, Frostig, and Soltanolkotabi]{tu2022learning}
S.~Tu, R.~Frostig, and M.~Soltanolkotabi.
\newblock Learning from many trajectories.
\newblock \emph{arXiv preprint arXiv:2203.17193}, 2022.

\bibitem[van Trees(2004)]{van2004detection}
H.~L. van Trees.
\newblock \emph{{Detection, Estimation, and Modulation Theory, Part I:
  Detection, Estimation, and Linear Modulation Theory}}.
\newblock John Wiley \& Sons, 2004.

\bibitem[Wagenmaker et~al.(2021)Wagenmaker, Simchowitz, and
  Jamieson]{wagenmaker2021task}
A.~Wagenmaker, M.~Simchowitz, and K.~Jamieson.
\newblock Task-optimal exploration in linear dynamical systems.
\newblock \emph{arXiv preprint arXiv:2102.05214}, 2021.

\bibitem[Zheng and Li(2020)]{zheng2020non}
Y.~Zheng and N.~Li.
\newblock Non-asymptotic identification of linear dynamical systems using
  multiple trajectories.
\newblock \emph{IEEE Control Systems Letters}, 5\penalty0 (5):\penalty0
  1693--1698, 2020.

\bibitem[Zhou et~al.(1996)Zhou, Doyle, and Glover]{zhou1996robust}
K.~Zhou, J.~Doyle, and K.~Glover.
\newblock \emph{Robust and Optimal Control}.
\newblock Feher/Prentice Hall Digital and. Prentice Hall, 1996.
\newblock ISBN 9780134565675.

\bibitem[Ziemann and Sandberg(2022)]{ziemann2022regret}
I.~Ziemann and H.~Sandberg.
\newblock Regret lower bounds for learning linear quadratic gaussian systems.
\newblock \emph{arXiv preprint arXiv:2201.01680. Manuscript in preparation},
  2022.

\bibitem[Ziemann et~al.(2022)Ziemann, Tsiamis, Sandberg, and
  Matni]{ziemann2022policy}
I.~Ziemann, A.~Tsiamis, H.~Sandberg, and N.~Matni.
\newblock How are policy gradient methods affected by the limits of control?
\newblock \emph{arXiv preprint arXiv:2206.06863. To appear at {CDC'22}}, 2022.

\end{thebibliography}

\ifnum\value{cdc}>0{}\else{
\appendix
\section{Proofs from Section~\ref{s: suboptimality lower bounds}: Excess Cost Lower Bound}

\begin{lemma}(Cauchy-Schwarz)
    \label{lem: matrix cauchy}
    For two sequences $a_1, \dots, a_n \in \R^n$ and $b_1, \dots, b_n \in \R^n$, 
    \[
        \sym \sum_{i=1}^n a_i b_i^\top \preceq \sum_{i=1}^n a_i a_i^\top + \sum_{i=1}^n b_i b_i^\top. 
    \]
\end{lemma}
\begin{proof}
    Express $A = \bmat{a_1 & \dots & a_n}$, and $B=\bmat{b_1 & \dots & b_n}$. The result follows by rearranging the inequality $0 \preceq (A-B)(A-B)^\top$. 
\end{proof}

\begin{lemma}   
    \label{lem: covariance overlap between learned and optimal controller}
    Suppose $T \geq \frac{16\lambda_{\min}(\Sigma_X)}{\norm{\Sigma_X}^2}$. Then under event $\calE$, 
    \[
        \norm{\Sigma_X^{-1/2}\Sigma_{\Theta}^{\hat K(\calZ)}\Sigma_X^{-1/2} - I}  \leq \frac{1}{2}.
    \]
\end{lemma}
\begin{proof}
    We first bound $\norm{\Sigma_X^{-1/2}\Sigma_{\Theta}^{\hat K(\calZ)}\Sigma_X^{-1/2} - I}$ in terms of the gap between $\hat K(\calZ)$ and $K(\Theta)$. In particular, we have that
\begin{align*}
    \norm{\Sigma_X^{-1/2}\Sigma_{\Theta}^{\hat K(\calZ)}\Sigma_X^{-1/2} - I} = \norm{\Sigma_X^{-1/2}(\Sigma_{\Theta}^{\hat K(\calZ)} - \Sigma_X)\Sigma_X^{-1/2}} \leq \frac{\norm{\Sigma_{\Theta}^{\hat K(\calZ)} - \Sigma_X}}{\lambda_{\min}(\Sigma_X)}.
\end{align*}
Note that 
\begin{align*}
    \Sigma_{\Theta}^{\hat K(\calZ)} & - \Sigma_X = %\frac{1}{T}\paren{ \sum_{t=0}^{T-1} (A+B\hat K)^t \Sigma_X \paren{(A+B\hat K)^t}^\top +
    \frac{1}{T} \sum_{t=0}^{T-1} \sum_{k=0}^{t-1} (A+B \hat K)^k \Sigma_W \paren{(A+B \hat K)^k}^\top -\Sigma_X \\
    &= \frac{1}{T} \sum_{t=0}^{T-1} \sum_{k=0}^{t-1} (A+B \hat K)^k \Sigma_W \paren{(A+B \hat K)^k}^\top - \frac{1}{T} \sum_{t=0}^{T-1} \sum_{k=0}^\infty (A+BK)^k \Sigma_W \paren{(A+BK)^k} \\
   % &= \frac{1}{T} \sum_{t=0}^{T-1} \sum_{k=0}^{t-1} (A+B \hat K)^k \Sigma_W \paren{(A+B \hat K)^t}^\top - \sum_{k=0}^\infty (A+BK)^k \Sigma_W \paren{(A+BK)^k} \\
    &= \frac{1}{T} \sum_{t=0}^{T-1} \paren{\sum_{k=0}^{t-1} (A+B \hat K)^k \Sigma_W \paren{(A+B \hat K)^k}^\top - \sum_{k=0}^{t-1} (A+BK)^k \Sigma_W \paren{(A+BK)^k}}\\
    &\qquad - \frac{1}{T} \sum_{t=0}^{T-1} \sum_{k=t}^{\infty} (A+BK)^k \Sigma_W \paren{(A+BK)^k}^\top.
\end{align*} 
Then 
\begin{align*}
     &\norm{\Sigma_{\Theta}^{\hat K(\calZ)} - \Sigma_X} \\
     &\leq \frac{1}{T}\norm{\sum_{t=0}^{T-1} \paren {\sum_{k=0}^{t-1} (A+B \hat K)^k \Sigma_W \paren{(A+B \hat K)^k}^\top - \sum_{k=0}^{t-1} (A+BK)^k \Sigma_W \paren{(A+BK)^k}} } \\
     &\qquad + \frac{1}{T} \norm{\sum_{t=0}^{T-1} \sum_{k=t}^{\infty} (A+BK)^k \Sigma_W \paren{(A+BK)^k}^\top} \\
     &\leq \frac{1}{T}\norm{\sum_{t=0}^{T-1} \paren{ \sum_{k=0}^{t-1} (A+B \hat K)^k \Sigma_W \paren{(A+B \hat K)^k}^\top - \sum_{k=0}^{t-1} (A+BK)^k \Sigma_W \paren{(A+BK)^k} }} \\
     &\qquad + \frac{1}{T} \norm{\sum_{t=0}^{\infty} (A+BK)^{t} \sum_{k=0}^{\infty} (A+BK)^k \Sigma_W \paren{(A+BK)^k}^\top \paren{(A+BK)^t}^\top} \\
      &\leq \frac{1}{T}\norm{\sum_{t=0}^{T-1} \paren{\sum_{k=0}^{t-1} (A+B \hat K)^k \Sigma_W \paren{(A+B \hat K)^k}^\top - \sum_{k=0}^{t-1} (A+BK)^k \Sigma_W \paren{(A+BK)^k} }} + \frac{\norm{\Sigma_X}^2}{T}.
\end{align*}
With $T \geq \frac{4 \norm{\Sigma_X}^2}{\lambda_{\min}(\Sigma_X)}$, the quantity $\frac{\norm{\Sigma_X}^2}{T}$ is upper bounded by $\frac{\lambda_{\min}(\Sigma_X)}{4}$. To bound the remaining term, we apply Lemma~\ref{lem: lyapunov perturbation} to show that under event $\calE$, 
\begin{align*}
    \frac{1}{T}\norm{\sum_{t=0}^{T-1} \paren{\sum_{k=0}^{t-1} (A+B \hat K)^k \Sigma_W \paren{(A+B \hat K)^t}^\top - \sum_{k=0}^{t-1} (A+BK)^k \Sigma_W \paren{(A+BK)^k} } }\\
    \leq 6 \mathcal{J}(A+BK)^2 \norm{B} \norm{A+BK} \norm{\Sigma_X}\norm{\hat K - K} \leq \frac{\lambda_{\min}(\Sigma_X)}{4}.
\end{align*}
This yields the inequality $\norm{\hat \Sigma_{\Theta}^{\hat K(\calZ)} - \Sigma_X} \leq \frac{\lambda_{\min}(\Sigma_X)}{2}$, as we needed to show.
\end{proof}

\begin{lemma}
    \label{lem: lyapunov perturbation}
     Given a controller $K$ that stabilizes the system $x_{t+1} = Ax_t + Bu_t + w_t$, and another controller $\hat K$ such that $\norm{K - \hat K} \leq \min \curly{\frac{\norm{A+BK}}{\norm{B}}, \frac{1}{6 \mathcal{J}(A+BK) \norm{B}  \norm{A+BK}}}$. Let $P_1^0 = P_2^0 = \Sigma_W$, and 
    \begin{align*}
        P_1^t &= (A+B \hat K) P_1^{t-1} (A+B\hat K)^\top + \Sigma_W \\
        P_2^t &= (A+BK)P_2^{t-1} \paren{A+BK}^\top + \Sigma_W. 
    \end{align*}
    Then 
     \begin{align*}
    \frac{1}{T}\norm{\sum_{t=0}^{T-1} P_1^t - P_2^t} \leq 6 \mathcal{J}(A+BK) \norm{B} \norm{A+BK} \norm{\Sigma_X}\norm{\hat K - K}.
    \end{align*}
\end{lemma}
\begin{proof}
    We have that
    \begin{align*}
        P_1^t - P_2^t &= (A+BK)(P_1^{t-1} - P_2^{t-1})(A+BK)^\top \\
        &\qquad + \sym(B(\hat K - K) P_1^{t-1} (A+BK)^\top)+ B(\hat K - K)^\top P_1^{t-1} (\hat K - K)^\top B^\top,
    \end{align*}
    where for a square matrix $M$, $\sym(M) = M + M^\top$. Therefore,
    \begin{align*}
        P_1^t - P_2^t &= \sum_{k=1}^{t-1} (A+BK)^k \bigg(\sym(B(\hat K - K) P_1^{k-1} (A+BK)^\top)\\ 
        &\qquad+ B(\hat K - K)^\top P_1^{k-1} (\hat K - K)^\top B^\top \bigg) \paren{(A+BK)^k}^\top.
    \end{align*}
    By the triangle inquality and submultiplicativity, we have
    \begin{align}
        \label{eq: lyap gap bound}
        \norm{P_1^t - P_2^t} &\leq \calJ(A+BK) \norm{B}\norm{\hat K - K} \norm{P_1^t}\paren{\norm{\hat K - K}\norm{B} + 2\norm{A+BK}} \\
        &\leq 3\calJ(A+BK) \norm{B}\norm{\hat K - K} \norm{P_1^t}\norm{A+BK},
    \end{align}
    where the first inequality leveraged the fact that $\norm{P_1^{k-1}}\leq \norm{P_1^t}$ for $k-1 \leq t$, and the last inequality follows from the fact that $\norm{\hat K - K} \leq \frac{\norm{A+BK}}{\norm{B}}$.
    Next, note that
    \begin{align*}
        \norm{P_1^t} &\leq \norm{P_2^t} + \norm{P_1^t - P_2^2} \leq \norm{P_2^t} + 3\calJ(A+BK) \norm{B}\norm{\hat K - K} \norm{P_1^t}\norm{A+BK},
    \end{align*}
    so 
    \begin{align*}
        \norm{P_1^t} &\leq \frac{\norm{P_2^t}}{1 - 3\calJ(A+BK) \norm{B}\norm{\hat K - K} \norm{A+BK}} \leq 2 \norm{P_2^t},
    \end{align*}
    where the last inequality follows from the fact that $\norm{\hat K - K} \leq \frac{1}{6 \mathcal{J}(A+BK) \norm{B}  \norm{A+BK}}$. Substituting this into \eqref{eq: lyap gap bound} provides the inequality in the Lemma statement. 
\end{proof}

\subsection{Proof of \Cref{lem: fisher bound}}

\label{s: pf of fisher bound}
\begin{proof}
For any vector $w$, 
    \begin{align*}
        w^\top &I_p(\theta) w= \E_{\theta} \sum_{n=1}^N \sum_{t=0}^{T-1} w^\top \paren{\dop_\theta \VEC \bmat{A(\theta) & B(\theta)}}^\top  \paren{ Z_{t,n} Z_{t,n}^\top \otimes \Sigma_W^{-1} }\paren{\dop_\theta \VEC \bmat{A(\theta) & B(\theta)}} w \\
        &\leq \frac{1}{\lambda_{\min}(\Sigma_W)} \E_{\theta} \sum_{n=1}^N \sum_{t=0}^{T-1} w^\top \paren{ \dop_\theta \VEC \bmat{A(\theta) & B(\theta)}}^\top  \paren{ Z_{t,n} Z_{t,n}^\top \otimes I } \paren{\dop_\theta \VEC \bmat{A(\theta) & B(\theta)}} w \\
        &= \frac{1}{\lambda_{\min}(\Sigma_W)} \E_{\theta} \sum_{n=1}^N \sum_{t=0}^{T-1} w^\top \paren{ \dop_\theta \VEC \bmat{A(\theta) & B(\theta)}}^\top  \VEC \paren{\VEC^{-1}  \paren{\paren{\dop_\theta \VEC \bmat{A(\theta) & B(\theta)}} w}  Z_{t,n} Z_{t,n}^\top}  \\
        &= \frac{1}{\lambda_{\min}(\Sigma_W)} \E_{\theta} \sum_{n=1}^N \sum_{t=0}^{T-1} \trace\bigg( \VEC^{-1}  \paren{\paren{\dop_\theta \VEC \bmat{A(\theta) & B(\theta)}} w} \\
        &\qquad\qquad \cdot Z_{t,n} Z_{t,n}^\top \paren{\VEC^{-1} \paren{ \paren{ \dop_\theta \VEC \bmat{A(\theta) & B(\theta)}}w}}^\top \bigg)
    \end{align*}
     where the $\VEC^{-1}$ operator maps a vector $v\in\R^{\dx d_k}$ to  a matrix $\VEC^{-1} v \in \R^{\dx \times d_k}$ as $\VEC^{-1} v = \bmat{v_{1:\dx} & v_{\dx+1:2\dx} & \dots & v_{(d_k-1)\dx+1:d_k\dx}}$. The second to last inequality follows from the vectorization identity, $ \VEC (XYZ) = (Z^\top \otimes X) \VEC(Y)$, and the last line follows from the identity $\trace(XY) = \VEC(X) \VEC(Y)^\top$. 
     
     Observe that the quantity $\VEC^{-1} \paren{ \paren{ \dop_\theta \VEC \bmat{A(\theta) & B(\theta)}}w}$ may be expressed as 
     \begin{align*}
        \VEC^{-1} & \paren{ \paren{ \dop_\theta \VEC \bmat{A(\theta) & B(\theta)}}w} \\
        & = \bmat{(\dop_\theta A_1(\theta)) w & \dots & (\dop_\theta A_{\dx}(\theta)) w & (\dop_\theta B_1(\theta)) w & \dots & (\dop_\theta B_{\du}(\theta)) w } \\
        & = \bmat{\VEC^{-1} \paren{ \paren{ \dop_\theta \VEC A(\theta)} w} & \VEC^{-1} \paren{ \paren{ \dop_\theta \VEC B(\theta)} w}}, \\
     \end{align*}
     where $A_i$ and $B_i$ denote the $i^\textrm{th}$ column of $A$ and $B$ respectively. Then the above bound may be expressed
     \begin{align*}
          w^\top &I_p(\theta) w \leq  \frac{1}{\lambda_{\min}(\Sigma_W)} \E_{\theta} \sum_{n=1}^N \sum_{t=0}^{T-1} \trace\bigg( \bmat{\VEC^{-1} \paren{ \paren{ \dop_\theta \VEC A(\theta)} w} & \VEC^{-1} \paren{ \paren{ \dop_\theta \VEC B(\theta)} w}} \\
        &\qquad\qquad \cdot Z_{t,n} Z_{t,n}^\top \bmat{\VEC^{-1} \paren{ \paren{ \dop_\theta \VEC A(\theta)} w} & \VEC^{-1} \paren{ \paren{ \dop_\theta \VEC B(\theta)} w}}^\top \bigg) \\
        &= \frac{1}{\lambda_{\min}(\Sigma_W)} \E_{\theta} \sum_{n=1}^N \sum_{t=0}^{T-1} \norm{ \VEC^{-1} \paren{ \paren{ \dop_\theta \VEC A(\theta)} w} X_{t,n} +  \VEC^{-1} \paren{ \paren{ \dop_\theta \VEC B(\theta)} w} U_{t,n}}_F^2 \\
        &\leq  \frac{2}{\lambda_{\min}(\Sigma_W)} \E_{\theta} \sum_{n=1}^N \sum_{t=0}^{T-1} \norm{ \VEC^{-1} \paren{ \paren{ \dop_\theta \VEC A(\theta)} w} X_{t,n}}_F^2 +  \norm{\VEC^{-1} \paren{ \paren{ \dop_\theta \VEC B(\theta)} w} U_{t,n}}_F^2 \\
        &= \frac{2}{\lambda_{\min}(\Sigma_W)} \E_{\theta} \sum_{n=1}^N \sum_{t=0}^{T-1} \trace\paren{\VEC^{-1} \paren{ \paren{ \dop_\theta \VEC A(\theta)} w} X_{t,n}  X_{t,n}^\top \paren{\VEC^{-1} \paren{ \paren{ \dop_\theta \VEC A(\theta)} w}}^\top} \\
        &\qquad +  \trace\paren{\VEC^{-1} \paren{ \paren{ \dop_\theta \VEC B(\theta)} w} U_{t,n} U_{t,n}^\top \paren{\VEC^{-1} \paren{ \paren{ \dop_\theta \VEC B(\theta)} w}}^\top }.
     \end{align*}
    We may pull the summations and the expectation inside the trace, and pull out the norms of ${\E_{\theta} \sum_{n=1}^N \sum_{t=0}^{T-1} X_{t,n} X_{t,n}^\top }$ and ${\E_{\theta} \sum_{n=1}^N \sum_{t=0}^{T-1} U_{t,n} U_{t,n}^\top }$ to arrive at the following bound.
    \begin{align*}
        w^\top I_p(\theta) w  &\leq \frac{2\norm{D_{\theta} \VEC A(\theta) w}^2}{\lambda_{\min}(\Sigma_W)}  \norm{\E_{\theta} \sum_{n=1}^N \sum_{t=0}^{T-1} X_{t,n} X_{t,n}^\top } + \frac{2 \norm{D_{\theta} \VEC B(\theta) w}^2}{\lambda_{\min}(\Sigma_W)} \norm{\E_{\theta} \sum_{n=1}^N \sum_{t=0}^{T-1} U_{t,n} U_{t,n}^\top } \\
         & \leq \frac{2\norm{D_{\theta} \VEC A(\theta) w}^2 + 4\norm{D_{\theta} \VEC B(\theta) w}^2 \norm{F}^2}{\lambda_{\min}(\Sigma_W)}  \norm{\E_{\theta} \sum_{n=1}^N \sum_{t=0}^{T-1} X_{t,n} X_{t,n}^\top } \\
         &\qquad \qquad + \frac{4\norm{D_{\theta} \VEC B(\theta) w}^2}{\lambda_{\min}(\Sigma_W)} \norm{\E_{\theta} \sum_{n=1}^N \sum_{t=0}^{T-1} \tilde U_{t,n} \tilde U_{t,n}^\top } \\
        & \leq \frac{2\norm{D_{\theta} \VEC A(\theta) w}^2 + 4\norm{D_{\theta} \VEC B(\theta) w}^2 \norm{F}^2}{\lambda_{\min}(\Sigma_W)}  \norm{\E_{\theta} \sum_{n=1}^N \sum_{t=0}^{T-1} X_{t,n} X_{t,n}^\top } \\
        &\qquad \qquad + 4\frac{\norm{D_{\theta} \VEC B(\theta) w}^2}{\lambda_{\min}(\Sigma_W)}  \sigma_{\tilde u}^2 N T,
    \end{align*}
    where the first inequality follows by applying Cauchy-Schwarz and submultiplicativity to bound $\norm{\E_{\theta} \sum_{n=1}^N \sum_{t=0}^{T-1} U_{t,n} U_{t,n}^\top } \leq 2\norm{F}^2\norm{\E_{\theta} \sum_{n=1}^N \sum_{t=0}^{T-1} X_{t,n} X_{t,n}^\top }+2\norm{\E_{\theta} \sum_{n=1}^N \sum_{t=0}^{T-1} \tilde U_{t,n} \tilde U_{t,n}^\top }$. The second inequality follows by using \eqref{eq:budgeteq} to bound $\norm{\E_{\theta} \sum_{n=1}^N \sum_{t=0}^{T-1} \tilde U_{t,n} \tilde U_{t,n}^\top }$. 
    To bound the quantity  $\norm{\E_{\theta} \sum_{n=1}^N \sum_{t=0}^{T-1} X_{t,n} X_{t,n}^\top }$, observe that by \Cref{lem: matrix cauchy},
    \begin{align*}
        &\E_{\theta} \sum_{n=1}^N \sum_{t=0}^{T-1} X_{t,n} X_{t,n}^\top \\
        &= \E_{\theta} \sum_{n=1}^N \sum_{t=0}^{T-1} \paren{\sum_{k=0}^{t-1} (A+BF)^{t-1-k} W_{k,n} + (A+BF)^{t-1-k} B \tilde U_{k,n}} \\
        & \qquad \cdot \paren{\sum_{k=0}^{t-1} (A+BF)^{t-1-k} W_{k,n} + (A+BF)^{t-1-k} B \tilde U_{k,n}}^\top \\
        &\preceq 2 \E_{\theta} \sum_{n=1}^N \sum_{t=0}^{T-1} \paren{\sum_{k=0}^{t-1} (A+BF)^{t-1-k} W_{k,n} }\paren{\sum_{k=0}^{t-1} (A+BF)^{t-1-k} W_{k,n}}^\top \\
        &\qquad + 2 \E_{\theta} \sum_{n=1}^N \sum_{t=0}^{T-1} \paren{\sum_{k=0}^{t-1} (A+BF)^{t-1-k} B \tilde U_{k,n} }\paren{\sum_{k=0}^{t-1} (A+BF)^{t-1-k} B \tilde U_{k,n}}^\top.
    \end{align*}
    The first term may be simplified using the properties of the sequence $W_k$:
    \begin{align*}
        &\E_{\theta} \sum_{n=1}^N \sum_{t=0}^{T-1} \paren{\sum_{k=0}^{t-1} (A+BF)^{t-1-k} W_{k,n} }\paren{\sum_{k=0}^{t-1} (A+BF)^{t-1-k} W_{k,n}}^\top \\
        &= N \sum_{t=0}^{T-1} \sum_{k=0}^{t-1} (A+BF)^{k} \Sigma_W \paren{(A+BF)^{k}}^\top \preceq NT \sum_{t=0}^{\infty} (A+BF)^{t} \Sigma_W \paren{(A+BF)^{t}}^\top.
    \end{align*}
    We will massage the second term to bound it using the input energy bound in \eqref{eq:budgeteq}:
    \begin{align*}
        &\norm{\sum_{t=0}^{T-1} \paren{\sum_{k=0}^{t-1} (A+BF)^{t-1-k} B \tilde U_{k,n} }\paren{\sum_{k=0}^{t-1} (A+BF)^{t-1-k} B\tilde U_{k,n}}^\top} \\
        &=  \norm{\sum_{t=0}^{T-1} \sum_{k=0}^{t-1} \sum_{j=0}^{t-1}  \paren{(A+BF)^{t-1-k} B \tilde U_{k,n} }\paren{(A+BF)^{t-1-j} B \tilde U_{j,n}}^\top} \\
        & \leq \sum_{t=0}^{T-1} \sum_{k=0}^{t-1} \sum_{j=0}^{t-1}  \norm{(A+BF)^{t-1-k} B }\norm{(A+BF)^{t-1-j} B} \norm{\tilde U_{j,n}}\norm{\tilde U_{k,n}} \\
        &\leq \sum_{t=0}^{T-1} \sum_{k=0}^{t-1} \sum_{j=0}^{t-1}  \norm{(A+BF)^{t-1-k} B }\norm{(A+BF)^{t-1-j} B} \frac{ \norm{\tilde U_{j,n}}^2 + \norm{\tilde U_{k,n}}^2}{2} \\
        &=\sum_{t=0}^{T-1} \sum_{k=0}^{t-1} \sum_{j=0}^{t-1}  \norm{(A+BF)^{t-1-k} B}\norm{(A+BF)^{t-1-j} B} \norm{\tilde U_{j,n}}^2 \\
        &\leq \sum_{t=0}^{T-1} \sum_{k=0}^{T-1} \sum_{j=0}^{t-1}  \norm{(A+BF)^{t-1-k} B }\norm{(A+BF)^{t-1-j} B} \norm{\tilde U_{j,n}}^2 \\
        &= \paren{\sum_{k=0}^{T-1} \norm{(A+BF)^{t-1-k} B}} \sum_{j=0}^{T-1} \norm{\tilde U_{j,n}}^2 \sum_{t=j+1}^{T-1} \norm{A^{t-1-j} B}  \\
        &\leq \paren{\sum_{k=0}^{T-1} \norm{(A+BF)^{t-1-k} B}}^2 \sum_{j=0}^{T-1} \norm{\tilde U_{j,n}}^2.
    \end{align*}
    Then 
    \begin{align*}
        &\norm{\E_{\theta} \sum_{n=1}^N \sum_{t=0}^{T-1} \paren{\sum_{k=0}^{t-1} (A+BF)^{t-1-k} B \tilde U_{k,n} }\paren{\sum_{k=0}^{t-1} (A+BF)^{t-1-k} B\tilde U_{k,n}}^\top } \\
        &\leq  \sigma_{\tilde u}^2 NT \paren{\sum_{t=0}^{\infty} \norm{(A+BF)^t B}}^2.
    \end{align*}
    Combining these results proves the statement. 
    
\end{proof}

\subsection{Proof of \Cref{thm: finite data bound}}
\label{s: proof of finite data bound}
\begin{proof}
    We must show that for all $\pi \in \Pi^{\mathsf{lin}}$, $\sup_{\theta'\in\calB(\theta,\e)} \mathsf{EC}_T^\pi(\theta') \geq \frac{G}{8 NT \bar L}$. Suppose that for some $\pi \in \Pi^\mathsf{lin}$, $\sup_{\theta' \in \calB(\theta, \e) }\mathsf{EC}_T^\pi(\theta') \leq \frac{G}{8 NT \bar L}$. We have by \Cref{lem: coordinate transformation} that this policy satisfies (under the assumption $T \geq \sup_{\theta'\in\calB(\theta,\e)} \frac{16 \norm{\Sigma_X(\theta')}}{\lambda_{\min}(\Sigma_X(\theta')}$ for case 2) 
    \begin{align*}
        %\label{eq: single policy lb}
        \sup_{\theta'\in\calB(\theta,\e)} \mathsf{EC}_T^\pi(\theta') \geq \frac{\tr\left(\Xi_{\theta,\e}  \E[\dop_\theta \VEC  K(\Theta) V \mathbf{1}_\calG] \E [\dop_\theta  \VEC K(\Theta) V 
         \mathbf{1}_\calG ]^\T \right)}{ \norm{V^\top \paren{\E \I_p(\Theta)+\J(\lambda)}V} }.
    \end{align*}
    The burn-in requirement $TN \geq \frac{\norm{J(\lambda)}}{\bar L}$ enables upper bounding $\norm{V^\top J(\lambda) V}$ by $TN \bar L$.  \Cref{lem: fisher bound} then allows us to upper bound the denominator in \eqref{eq: single policy lb} by $2TN \bar L$. 

    To remove the indicators from the lower bound in \eqref{eq: single policy lb}, we take an infimum over $\tilde \theta, \theta' \in \calB(\theta,\e)$ to lower bound the numerator in \eqref{eq: single policy lb} by $\bfP[\calG]^2 G$. For case 1, we immediately have $\bfP[\calG]^2=  \bfP[\Omega]^2 = 1$. For case 2, we leverage the assumptions that the prior is small and that the burn-in time is satisfied to show that $\bfP[\calG]^2 =\bfP[\calE]^2 \geq \frac{1}{4}$. To do so, we must show that $\sup_{\theta' \in \calB(\theta,\e)} \norm{\hat K(\calZ) - K(\theta')}$ is small with high probability. Note that 
     \begin{align*}
        \sup_{\theta' \in \calB(\theta,\e)} \norm{\hat K(\calZ) - K(\theta')} &\leq  \sup_{\theta_2 \in \calB(\theta,\e)} \inf_{\theta_1 \in \calB(\theta,\e)}  \norm{\hat K(\calZ) - K(\theta_1)} + \norm{K(\theta_1) - K(\theta_2)} \\
        &\leq \inf_{\theta_1 \in \calB(\theta,\e)}  \norm{\hat K(\calZ) - K(\theta_1)} + \sup_{\theta_1, \theta_2 \in \calB(\theta,\e)}  \norm{K(\theta_1) - K(\theta_2)}.
    \end{align*}
    
    By Propositions 1 and 2 of \cite{mania2019certainty}, 
        $\sup_{\theta_1, \theta_2 \in \calB(\theta,\e)}  \norm{K(\theta_1) - K(\theta_2)} \leq c_1 \epsilon$,
    so long as $\e \leq c_2$. 
    Therefore, if $\varepsilon \leq \min \curly{\frac{\alpha}{2 c_1}, c_2}$, then $\sup_{\theta_1, \theta_2 \in \calB(\theta,\e)}  \norm{K(\theta_1) - K(\theta_2)} \leq \alpha/2$. 
    
    For the remaining term, note that for any $\calZ$, $\Theta$ we have $\inf_{\theta_1 \in \calB(\theta,\e)}  \norm{\hat K(\calZ) - K(\theta_1)} \leq \norm{\hat K(\calZ) - K(\Theta)} = \sqrt{\norm{\hat K(\calZ) - K(\Theta)}^2}$. This quantity may in turn be bounded by the excess cost:
    \begin{align*}
        \norm{\hat K(\calZ) - K(\Theta)}^2 \leq  \frac{\tr\paren{(\hat K(\calZ) - K(\Theta))^\top \Psi(\Theta) (\hat K(\calZ)-K(\Theta) \Sigma_{\Theta}^{\hat K(\calZ)}  )}}{ \lambda_{\min}(\Sigma_{\Theta}^{\hat K(\calZ)}) \lambda_{\min}(\Psi(\Theta))} \leq \frac{\mathsf{EC}_T^{\pi(\cdot; \calZ)}(\Theta)}{\lambda_{\min}(\Sigma_W) \lambda_{\min}(R)},
    \end{align*}
    where $\mathsf{EC}_T^{\pi(\cdot; \calZ)}(\Theta)$ is the conditional excess cost given $\calZ$ and $\Theta$:
    \[
        \mathsf{EC}_T^{\pi(\cdot; \calZ)}(\Theta) =\frac{1}{T} \E^\pi \brac{\sum_{t=0}^{T-1}  (U_t- K(\Theta)X_t )^\T \Psi(\Theta)  (U_t-K(\Theta)X_t ) \vert \calZ, \Theta}.
    \]
    By Markov's inequality, we have that
    \begin{align*}
        \bfP[\mathsf{EC}_T^{\pi(\cdot; \calZ)}(\Theta) > \lambda_{\min}(\Sigma_W) \lambda_{\min}(R) \alpha^2/4] \leq \frac{4 \E[\mathsf{EC}_T^{\pi(\cdot; \calZ)}(\Theta)]}{\lambda_{\min}(\Sigma_W) \lambda_{\min}(R) \alpha^2} \leq \frac{4 \sup_{\theta' \in \calB(\theta,\e)} \mathsf{EC}_T^\pi(\theta')}{ \lambda_{\min}(\Sigma_W) \lambda_{\min}(R) \alpha^2}.
    \end{align*}
    By our assumption at the beginning of the proof that $\sup_{\theta' \in \calB(\theta, \e) }\mathsf{EC}_T^\pi(\theta') \leq \frac{G}{8 NT \bar L}$, we see that as long as the burn-in requirement is satisified, 
    \[
        \bfP[\mathsf{EC}_T^{\pi(\cdot; \calZ)}(\Theta) \leq \lambda_{\min}(\Sigma_W) \lambda_{\min}(R) \alpha^2/4] \geq \frac{1}{2}.
    \]
    As a result, $\bfP[\calG]\geq \frac{1}{2}$. 
    
    This in turn implies that for both case 1 and case 2, 
    \begin{align*}
        \sup_{\theta'\in\calB(\theta,\e)} \mathsf{EC}_T^\pi(\theta') \geq \frac{G}{8 TN\bar L}.
    \end{align*}
    Therefore, for all $\pi \in \Pi^{\mathsf{lin}}$, the above lower bound is satisfied, and 
    \begin{align*}
        \mathcal{EC}^{\mathsf{lin}}_T(\theta,\e) = \inf_{\pi \in \Pi^\mathsf{lin}} \sup_{\theta' \in \calB(\theta,\e)} \mathsf{EC}_T^\pi(\theta') \geq \frac{G}{8 TN\bar L}.
    \end{align*}

    %\Tasos{it seems to me that while proving the lower bound on the probability event there is an interesting tension. Low regret leads to small K error leading to large regret. it might be worth to comment on that more} 
\end{proof}

\subsection{Proof of \Cref{cor: asymptotic lower bound}}
\label{s: asymptotic lb proof}

\begin{proof}
For the burn-in requirements to be satisfied asymptotically, it is necessary to select a prior $\lambda$ such that 
$\norm{J(\lambda)}$ grows slower than $N$. By our assumption in the corollary statement that $\e = N^{-\alpha}$ for $\alpha \in (0,\frac{1}{2})$, it is sufficient to show that $\norm{J(\lambda)} \leq \frac{c}{\e^2}$ for a constant $c$ which does not depend on $N$. To do so, let $\rho$ be a smooth density supported on $B(0,1)$, and let $\lambda(\tilde \theta)= \frac{1}{\e} \rho\paren{\frac{\tilde \theta - \theta}{\eps}}$. By the chain rule of differentiation and a change of variables, we have that
\[
    J(\lambda) = \frac{1}{\e^2} J(\rho). 
\]
\end{proof}

%\Tasos{proof of corollary 2.1 is missing? how do we bound the $J(\theta)$? Why should we take alpha less than 1/2? Is there a reason why $J(\theta)\le O(\epsilon^{-2})$ (I tried to find an explanation in Ingvar's paper but I couldn't find it)? Do we assume that we scale linearly (and then we shift) a fixed distribution $\rho$ around $B(0,1)$, e.g. $\lambda(\tilde{\theta})=\frac{1}{\epsilon}\rho(\frac{\tilde{\theta}-\theta}{\epsilon})$, where now lambda is defined as desired in $B(\theta,\epsilon)$? this should give that $J(\theta)=\epsilon^{-2}\times \textrm{const}$ I guess. I don't think this obvious. If we decide to scale with sth else, e.g. $\epsilon^2$, then this no longer holds.}

\subsection{Proof of \Cref{cor: global minimax}}
    \label{s: global minimax pf}
    
     We argue as in the proof of \Cref{cor: asymptotic lower bound} with $V = \bmat{0 \\ 1}$. In this perturbation direction, the lower bound evaluates to $\frac{b^2 P +1}{32}\frac{\partial K}{\partial b}(a,b)^2$. Considering $a = 1-\gamma$ and $b =\gamma$ for $0 < \gamma < 1$ and taking the limit as $\gamma\to 0$ results in the lower bound of $\infty$. 

    In particular, in this setting the solution to the Riccati equation evaluates to 
    \begin{align}
    \label{eq: scalar riccati} 
    P = \frac{ -\sqrt{(1-a^2)^2 - 2(1-a^2)b^2 + b^4} + -\sqrt{(1-a^2)^2 + 2(1+a^2)b^2 + b^4}}{2b^2}.
    \end{align}

    Next note that by the product rule, 
    \begin{align}
        \label{eq: controller derivative}
        \frac{dK}{db}(a,b) = \frac{b^2 P^2a - aP - b \frac{dP}{db}(a,b)a }{(b^2 P + 1)^2}.
    \end{align}
    Using \eqref{eq: scalar riccati}, it can be shown that the denominator of the above quantity converges to $1$. 
    We show that the numerator diverges to $\infty$ as $\gamma \to \infty$. To see that this is so, note that Section C.2.1 of \cite{simchowitz2020naive} may be used to show
    \[
        \frac{dP}{db}(a,b) = \dlyap(a_{cl}, 2 a_{cl} P K),
    \]
    where $a_{cl} = a+bK$. We have that $K = -bPa_{cl}$. Using these results, the numerator of \eqref{eq: controller derivative} may be simplified to 
    \begin{align*}
        b^2 P^2 a - Pa + \frac{2b^2 P^2 a_{cl}^2 a}{1-a_{cl}^2} &= P\paren{b^2 P a - a +2 \frac{b^2 P a_{cl}^2 a}{1-a_{cl}^2}}.
    \end{align*}
    It may be shown using Mathematica that 
    \[
    \lim_{\gamma\to\infty} \paren{b^2 P a - a +2 \frac{b^2 P a_{cl}^2 a}{1-a_{cl}^2}} = -\frac{\sqrt{2}}{2}. 
    \]
    Meanwhile, $P$ diverges to $\infty$.

\section{Proofs from \Cref{s: consequences of lower bound}}
\label{s: proofs of consequences} 

\subsection{Proof of \Cref{prop: dimensional dependence}}
\begin{proof}
        In this case the quantity $G$ in  \Cref{cor: asymptotic lower bound} with $\Gamma = \frac{1}{2}\Sigma_X$ may be lower bounded as 
   \begin{align*}
        &G = \frac{1}{2}\tr\big( (\Sigma_X \otimes \Psi)  \dop_\theta \VEC  K(\theta )V \paren{\dop_\theta  \VEC K(\theta) V }^\T \big) \\
        &\overset{(i)}{=} \frac{1}{2}\sum_{i \in [\dx \du]} \trace\paren{(\Sigma_X \otimes \Psi) d_{v_i} \VEC K(\theta) d_{v_i} \VEC K(\theta)^\top } \\
        & \geq \frac{1}{2}\dx\du \inf_{v \in \boldsymbol \Delta} \trace\big((\Sigma_X \otimes \Psi)  d_{v} \VEC K(\theta) d_{v} \VEC K(\theta)^\top \big) \\
        &\overset{(ii)}{=}  \frac{1}{2}\dx\du \inf_{v \in \boldsymbol \Delta} \trace\paren{\Psi d_{v} K(\theta) \Sigma_X d_{v}  K(\theta)^\top } \\
        &\overset{(iii)}{=} \frac{1}{2}\dx\du \inf_{\Delta : \norm{\bmat{-\Delta K & \Delta}}_F = 1} \trace\big(\Psi^{-1} \Delta^\top  P A_{cl} \Sigma_X A_{cl}^\top P \Delta \big)\\
        &\geq \frac{1}{2}\dx\du \inf_{\Delta : \norm{\bmat{-\Delta K & \Delta}}_F = 1} \trace\paren{\Delta^\top \Delta }   \frac{\lambda_{\min}(P)^2 \lambda_{\min}(\Sigma_X - \Sigma_W)}{\norm{\Psi}},
    \end{align*} 
    where $(i)$ follows from expanding the matrix $V = \bmat{v_1 & \dots & v_{\dx\du}}$ and the fact that $d_v \VEC K(\theta) = \dop_\theta \VEC K(\theta) v$, $(ii)$ follows from the identities $
        \VEC(XYZ) = (Z^\top \otimes X) \VEC(Y)$ and $\trace(\VEC(X) \VEC(Y)^\top) = \trace(XY^\top)$, and $(iii)$ follows from the directional derivative in \eqref{eq: polderman derivative}. The inequality results from the observations that $A_{cl}\Sigma_X A_{cl}^\top = \Sigma_X - \Sigma_W$.  
    We have that $\trace\paren{\Delta^\top \Delta} = \norm{\Delta}_F^2$.  Then using the fact that $1 = \norm{\bmat{-\Delta K & \Delta }}_F^2 \leq \norm{\bmat{-K & I}}^2 \norm{\Delta}_F^2$, we have $\norm{\Delta}_F^2 \geq \frac{1}{\norm{\bmat{-K & I}}^2}$.  
    The quantity $\tilde L$ upper bounds $L(\theta)$ in \Cref{cor: asymptotic lower bound} using the fact that $\norm{\dop_\theta \VEC \bmat{A(\theta) & B(\theta)} v} = \norm{v} = 1$ for all $v \in \boldsymbol \Delta$. 
    % The denominator $L$ from \Cref{cor: asymptotic lower bound} may be bounded by $\bar L$, as in \Cref{prop: system theoretic}. 
    
       %\norm{\dop_{\theta} \VEC K(\theta) V}_F^2 \geq d_x d_u \inf_{\Delta} \norm{d_\Delta K(\theta)}_F^2 \geq \frac{d_x d_u \sigma_{\min}(A_{cl})^2 \sigma_{\min}(P)^2}{\norm{B^\top P B + R}^2 \norm{\bmat{-K & I}}^2}. 
    \end{proof}

\subsection{Proof of \Cref{prop: exponential}}
\begin{proof}
We apply \Cref{cor: asymptotic lower bound}. For the perturbation $V = \VEC \bmat{0 & B/\norm{B}_F}$, $L(\theta)$ reduces to $8 \sigma_{\tilde u}^2$. Additionally, the directional derivative becomes
\ifnum\value{cdc}>0{
\begin{align*}
    d_V K(\theta) &= \frac{1}{\norm{B}_F}\Psi^{-1} (B^\top P A_{cl} + B^\top P B K \\ &  + B^\top \texttt{dlyap}(A_{cl}, \textrm{sym}(A_{cl}^\top P B K)) A_{cl}).
\end{align*}
}\else{
\begin{align*}
    d_V K(\theta) = \frac{1}{\norm{B}_F}(B^\top P B + R)^{-1} (B^\top P A_{cl} + B^\top P B K + B^\top \texttt{dlyap}(A_{cl}, \textrm{sym}(A_{cl}^\top P B K)) A_{cl}).
\end{align*}}\fi
As $R = 1$, we may express $K = -B^\top P A_{cl}$. We also note that $\norm{B}_F = 1$. Therefore,  
\ifnum\value{cdc}>0{
\begin{equation}
\label{eq: integrator derivative}
\begin{aligned}
    d_V K(\theta) &= \Psi^{-1} \big(B^\top P A_{cl} - B^\top P B B^\top P A_{cl}  \\
    &\qquad - 2 B^\top \texttt{dlyap}(A_{cl}, A_{cl}^\top P B B^\top P A_{cl}) A_{cl} \big) \\
    &= \Psi^{-1} \big(-  B^\top \big(P (P_{\dx\dx} - 1) \\
    &\qquad +  2\texttt{dlyap}(A_{cl}, A_{cl}^\top P B B^\top P A_{cl})\big)A_{cl} \big),
\end{aligned}
\end{equation}
}\else{
\begin{equation}
\label{eq: integrator derivative}
\begin{aligned}
    d_V K(\theta) &= (B^\top P B + R)^{-1} (B^\top P A_{cl} - B^\top P B B^\top P A_{cl}  - 2 B^\top \texttt{dlyap}(A_{cl}, A_{cl}^\top P B B^\top P A_{cl}) A_{cl}) \\
    &= (B^\top P B + R)^{-1} (-  B^\top \paren{P (P_{\dx\dx} - 1) +  2\texttt{dlyap}(A_{cl}, A_{cl}^\top P B B^\top P A_{cl})}A_{cl} ),
\end{aligned}
\end{equation}}\fi
where the second equality follows by noting that the quantity $B^\top P B$ in the second term is equal the lower right scalar entry of $P$, denoted by $P_{\dx\dx}$ and then pulling out $-B^\top$ on the left, and $A_{cl}$ on the right.
Now, to lower bound $G$ from \Cref{cor: asymptotic lower bound}, we begin with the bound on $G$ in \eqref{eq: F LB exp},
\ifnum\value{cdc}>0{
\begin{align*}
    G & \geq \tr(\Psi d_V K(\theta) d_V K(\theta)^\top)  \\
    &\overset{(i)}{\geq} \frac{1}{(P_{nn}+1)} || B^\top  \big(P (P_{\dx\dx} - 1) \\ &\qquad+  \texttt{dlyap}(A_{cl}, A_{cl}^\top P B B^\top P A_{cl})\big)A_{cl} ||_F^2 \\
    &\overset{(ii)}{\geq} \frac{(P_{\dx\dx}-1)^2}{(P_{\dx\dx}+1)} \norm{B^\top  P A_{cl}}_F^2  = \frac{(P_{\dx\dx}-1)^2}{(P_{\dx\dx}+1)} \norm{K}^2.
\end{align*}
}\else{
\begin{align*}
    G & \geq \tr((B^\top P B +R) d_V K(\theta) d_V K(\theta)^\top)  \\
    &\overset{(i)}{\geq} \frac{1}{(P_{nn}+1)} \| B^\top  \paren{P (P_{\dx\dx} - 1) +  \texttt{dlyap}(A_{cl}, A_{cl}^\top P B B^\top P A_{cl})}A_{cl} \|_F^2 \\
    &\overset{(ii)}{\geq} \frac{(P_{\dx\dx}-1)^2}{(P_{\dx\dx}+1)} \norm{B^\top  P A_{cl}}_F^2  = \frac{(P_{\dx\dx}-1)^2}{(P_{\dx\dx}+1)} \norm{K}^2.
\end{align*}}\fi
Here, $(i)$ follows by substituting $d_V K(\theta)$ for the expression in \eqref{eq: integrator derivative}, pulling out the remaining $\Psi^{-1} = (P_{\dx\dx}+1)^{-1}$ term out of the trace, and writing the trace as a Frobenius norm. The inequality $(ii)$ follows by the fact that $A_{cl} \paren{P (P_{\dx\dx} - 1) +  \texttt{dlyap}(A_{cl}, A_{cl}^\top P B B^\top P A_{cl})} A_{cl}^\top \succeq A_{cl} P (P_{\dx\dx}-1) A_{cl}^\top$. 
We next show that $\norm{K}$ does not scale inversely with the exponential of the system dimension. In particular, we show $\norm{K} \geq \frac{\rho}{2}$.
To see that this is so, note that
\begin{align*}
    \norm{K} &= \norm{B^\top P A_{cl}} \geq \norm{B^\top A_{cl}} \\
    & = \norm{\bmat{0 & \hdots & 0 & \rho} + K} \geq |K_{\dx} + \rho|. 
\end{align*}
Thus $\norm{K} \geq \max\curly{|K_{\dx}|, |K_{\dx}+\rho|} \geq \rho/2$. Therefore, our excess cost is lower bounded as 
    \begin{align*}
        \liminf_{N \to \infty} \sup_{\theta' \in \mathcal{B}(\theta, N^{-\alpha}) } N \mathsf{EC}_T^\pi(\theta')   \geq \frac{G}{8 T  L(\theta)} \geq \frac{\rho^2}{64T  \sigma_{\tilde u}^2} \frac{(P_{\dx\dx}-1)^2}{(P_{\dx\dx}+1)}
    \end{align*}
    By Lemma E.4 in \cite{tsiamis2022learning}, $P_{\dx\dx} \geq  4^{\dx-2}$. Then for $\dx \geq 3$, the quantity above is lower bounded by $\frac{\rho^2}{256T \sigma_{\tilde u}^2} 4^{\dx-2}$.
\end{proof}

\subsection{Proof of \Cref{prop: system theoretic}}

\begin{proof}
    Begin with the result of \Cref{cor: asymptotic lower bound} so that 
    \[
        \liminf_{N \to \infty} \sup_{\theta' \in \mathcal{B}(\theta, N^{-\alpha}) } N \mathsf{EC}_T^\pi(\theta')   \geq \frac{G}{8 TL(\theta)}.
    \]
    The denominator $L(\theta)$ may be bounded by $\tilde L$, as in \Cref{prop: dimensional dependence}. 
    % The quantity $\bar L$ upper bounds $L$ in \Cref{cor: asymptotic lower bound} using the fact that $\norm{\dop_\theta \VEC \bmat{A(\theta) & B(\theta)} v} = \norm{v} = 1$ for all $v \in \boldsymbol \Delta$. 
    To bound $G$, begin with the expression in \eqref{eq: G system theoretic}. By the fact that $A_{cl} \Sigma_X A_{cl}^\top = \Sigma_X - \Sigma_W$, this expression is lower bounded by 
    \ifnum\value{cdc}>0{
    \begin{align*}
        % &\tr((B^\top P B+R)^{-1} B^\top \dlyap\paren{A_{cl}, P} A_{cl} \Sigma_X A_{cl}^\top \dlyap\paren{A_{cl}, P} B) \\
        &\lambda_{\min}(\Sigma_X - \Sigma_W) \tr(\Psi^{-1} B^\top \dlyap\paren{A_{cl}, P}^2 B) \\
        &= \lambda_{\min}(\Sigma_X- \Sigma_W) \tr(\Psi^{-1} B^\top P^{1/2} P^{-1/2} \\ &\qquad \cdot\dlyap\paren{A_{cl}, P}^2 P^{-1/2} P^{1/2} B).
    \end{align*}
    }\else{
    \begin{align*}
        % &\tr((B^\top P B+R)^{-1} B^\top \dlyap\paren{A_{cl}, P} A_{cl} \Sigma_X A_{cl}^\top \dlyap\paren{A_{cl}, P} B) \\
        &\lambda_{\min}(\Sigma_X - \Sigma_W) \tr((B^\top P B+R)^{-1} B^\top \dlyap\paren{A_{cl}, P}^2 B) \\
        &= \lambda_{\min}(\Sigma_X- \Sigma_W) \tr((B^\top P B+R)^{-1} B^\top P^{1/2} P^{-1/2} \dlyap\paren{A_{cl}, P}^2 P^{-1/2} P^{1/2} B).
    \end{align*}}\fi
    Now, let $P^{1/2} B = U \Sigma V^\top$. Then $B^\top P B = V \Sigma^2 V^\top = V \Lambda_{B^\top P B} V$. The trace above then reduces to 
    \begin{align*}
        &\tr(\Sigma^2 (\Sigma^2 + \Lambda_R)^{-1} U^\top P^{-1/2}\dlyap\paren{A_{cl}, P}^2 P^{-1/2}  U) \\
        &\geq \inf_{i \in [\du]} \frac{\Sigma_{ii}^2 \trace(U^\top P^{-1/2}\dlyap\paren{A_{cl}, P}^2 P^{-1/2} U)}{\Sigma_{ii}^2 + \Lambda_{R, ii}}  \\
        &\geq \inf_{i \in [\du]} \frac{\Sigma_{ii}^2 \sum_{j=1}^{du} \lambda_{n-j} (P^{-1/2}\dlyap\paren{A_{cl}, P}^2 P^{-1/2})}{\Sigma_{ii}^2 + \Lambda_{R, ii}}  \\
        &\geq \inf_{i \in [\du]} \frac{\Sigma_{ii}^2 \sum_{j=1}^{du} \lambda_{n-j} (\dlyap\paren{A_{cl}, P}) }{\Sigma_{ii}^2 + \Lambda_{R, ii}} .
    \end{align*}
\end{proof}
}\fi

\end{document}